\documentclass[aps,reprint,superscriptaddress,longbibliography]{revtex4-1}

\pdfoutput=1

\usepackage{times}
\usepackage{mathtools}
\usepackage{amssymb}
\usepackage{mathrsfs}
\usepackage{amsthm}

\usepackage{xcolor}
\usepackage[colorlinks, citecolor={blue!70!black}, urlcolor={blue!70!black}, linkcolor={red!70!black},hyperindex,breaklinks]{hyperref}

\usepackage{graphicx}
\usepackage{float}
\usepackage{subfigure}
\usepackage{dashrule}

\newcommand{\beq}{\begin{equation}}
\newcommand{\eeq}{\end{equation}}

\DeclarePairedDelimiter{\abs}{\lvert}{\rvert}
\DeclarePairedDelimiter{\norm}{\lVert}{\rVert}
\DeclareMathOperator{\csch}{csch}
\DeclareMathOperator{\sech}{sech}

\theoremstyle{definition}
\newtheorem{lem}{Lemma}

\begin{document}

\title{Almost Perfect Metals in One Dimension}

\author{Chaitanya Murthy}
\affiliation{Department of Physics, University of California, Santa Barbara, CA 93106, USA}

\author{Chetan Nayak}
\affiliation{Department of Physics, University of California, Santa Barbara, CA 93106, USA}
\affiliation{Microsoft Quantum, Station Q, University of California, Santa Barbara, CA 93106, USA}

\begin{abstract}

We show that a one-dimensional quantum wire with as few as 2 channels of interacting fermions can host metallic states of matter that are stable against all perturbations up to $q^\text{th}$-order in fermion creation/annihilation operators for any fixed finite $q$.
The leading relevant perturbations are thus complicated operators that are expected to modify the physics only at very low energies, below accessible temperatures. 
The stability of these non-Fermi liquid fixed points is due to strong interactions between the channels, which can (but need not) be chosen to be purely repulsive. 
Our results might enable elementary physical realizations of these phases.

\end{abstract}

\maketitle

% Main text
% =======================================================================

\textbf{Introduction.}
Metallic states of matter are gapless and often unstable to either insulating behavior or superconductivity. 
This is especially true in one-dimensional systems,
where the localizing effects of disorder are particularly strong \cite{Lee1985}.
For a single channel (i.e.~a single propagating mode of each chirality at the Fermi energy), 
disorder-induced localization can only be avoided when the interaction is strongly attractive, 
while proximity-induced superconductivity can only be avoided when it is strongly repulsive \cite{Giamarchi2003}.
The situation is more complicated---and much more interesting---when there are multiple channels. 
We will show that, surprisingly, even for $N=2$ channels, it is possible to have a metallic state that is stable against all perturbations up to $q^\text{th}$-order in fermion creation/annihilation operators for any fixed finite $q$ (but not $q = \infty$), which we call \emph{(absolute) $q$-stability}.

Gapless phases of interacting fermions in one dimension are described at low energies by Luttinger liquid (LL) theory \cite{Haldane1981}.
They exhibit a remarkable and universal phenomenology that distinguishes them from Fermi liquids, but this is often obscured in experiments due to dimensional crossover, ordering, or localization \cite{Giamarchi2003}.
Thus, a physically realizable stable LL is not only interesting as a matter of principle, but also for the practical reason that it would provide a useful experimental platform to study non-Fermi liquid physics.

For $N=\infty$, it was shown two decades ago in Refs.~\cite{Kivelson1998,Golubovic1998,OHern1998,OHern1999, Emery2000,Vishwanath2001,Mukhopadhyay2001,Sierra2002} that there exist ``sliding Luttinger liquid'' phases which are stable against many, but not all, low-order perturbations; it was argued that the relevant perturbations are likely to have small bare values 
\footnote{We restrict attention in this paper to systems with short-ranged interactions. Long-ranged interactions can also stabilize a Luttinger liquid against a $2k_F$ potential and disorder \cite{Dora2016}}. 
More recently, it was discovered that it is possible for a one-dimensional metal to be stable against \emph{all} non-chiral perturbations (without restriction on the order) \cite{Plamadeala2014}. 
An explicit construction was given for $N = 23$ which exploited the properties of integral quadratic lattices. 
The present work shows that a slight relaxation of the condition of complete stability to the weaker condition of $q$-stability brings the required number of channels down from $23$ to $2$, thereby greatly increasing the chances of experimental realization.

The basic observation underlying the results of this paper and of Ref.~\cite{Plamadeala2014} is that the possible perturbations of an $N$-channel LL can be represented as lattice points in a fictitious $2N$-dimensional space equipped with two different metrics: 
the mixed-signature $(N,N)$ metric $\mathrm{diag}(-\mathbb{I}_N , \mathbb{I}_N)$ and the Euclidean metric $\mathbb{I}_{2N}$,
where $\mathbb{I}_N$ is the $N \times N$ identity matrix.
The mixed-signature interval from the origin to a lattice point measures the chirality of the associated perturbation, while the Euclidean interval measures its scaling dimension; points sufficiently far from the origin are irrelevant in the renormalization group (RG) sense.
The lattice is naturally graded into ``shells'' consisting of perturbations of a given order; low-order perturbations belong to the inner shells.
The effect of interactions is to deform the lattice by an $SO(N,N)$ transformation 
\footnote{
The Lie group $SO(N,N)$ consists of all matrices $A \in \mathbb{R}^{2N \times 2N}$ that satisfy $AKA^T = K$ and $\det A = 1$, where $K = \mathrm{diag}(-\mathbb{I}_N, \mathbb{I}_N)$.
}.
For $N=1$, the deformation is a Lorentz boost that is ``aligned'' with the lattice; 
such a boost unavoidably pulls one of the innermost lattice points closer to the origin, 
enhancing the susceptibility of the system to either localization or induced superconductivity.
For $N \geq 2$, on the other hand, the boosts can be ``misaligned'' with the lattice planes in such a way that all lattice points in the innermost $q$ shells are pushed away from the origin, making the corresponding perturbations irrelevant.

Remarkably, these absolutely $q$-stable phases can occur even for purely repulsive interactions. 
Two-channel repulsive LLs can occur in a number of different contexts. 
One simple example, with sufficient generality to permit the phases described here, is a single-spinful-channel quantum wire with strong spin-orbit coupling. 
In this case, the two Fermi points of each chirality have different Fermi momenta and velocities, and the interactions between the densities at the different Fermi points are not excessively constrained by symmetries.
Our construction shows that, for any fixed finite $q$, there exist purely repulsive local interactions for which such a metallic state is absolutely $q$-stable.

\vspace{0.5em}
\textbf{Model and Definitions.}
Consider a system of interacting fermions in a 1D quantum wire.
At low energies, the effective theory of the system involves $2N$ chiral spinless Dirac fermions $\psi_I$, where $\psi_I^{\dagger}$ ($\psi_{I+N}^{\dagger}$) creates a right-moving (left-moving) excitation about the Fermi point $k_{F,I}$ ($k_{F,I+N}$), with Fermi velocity $v_I > 0$ ($v_{I+N} < 0$).
The index $I$ distinguishes different bands, accounting for both spin and quantization of the transverse motion.
The effective action is given by $S_{\text{eff}}=S_0 + S_{\text{int}} + S_{\text{pert}}$, where
\beq
S_0 + S_{\text{int}} = \int dt \, dx \hspace{1pt}  
\Bigl[ \psi_I^{\dagger} \hspace{1pt} i(\partial_t + v_I \partial_x) \hspace{1pt} \psi^{\phantom{\dagger}}_I 
- \, U_{IJ} \hspace{1pt} \rho_I \hspace{1pt}  \rho_J \Bigr] .
\eeq
Here, the indices $I,J$ are implicitly summed from $1$ to $2N$, $\rho_I \equiv \psi^{\dagger}_I \psi^{\phantom{\dagger}}_I$, and the real symmetric $2N \times 2N$ matrix $U$ parametrizes all density-density interactions.
All other interaction terms, as well as any quadratic terms
accounting for dispersion nonlinearities, are packaged into $S_{\text{pert}}$.
If the system is perturbed in any way, 
for instance by introducing disorder or by proximity-coupling the wire to an external 3D superconductor, 
the appropriate terms are also included in $S_{\text{pert}}$.

The first part of the action, $S = S_0 + S_{\text{int}}$, describes a gapless $N$-channel Luttinger liquid.
$S$ can be treated non-perturbatively via the method of bosonization \cite{VonDelft1998}.
Introducing a chiral boson $\phi_I$ for each chiral fermion $\psi_I$, we obtain the bosonic representation
\beq
\label{eq:S_LL}
S = \frac{1}{4\pi} \int dt \, dx \, \Big[ K_{IJ} \partial_t \phi_I \partial_x \phi_J - V_{IJ} \partial_x \phi_I \partial_x \phi_J \Big] ,
\eeq
with $K = \mathrm{diag}(-\mathbb{I}_N , \mathbb{I}_N)$ and $V_{IJ} = \abs{v_I} \delta_{IJ} + \frac{1}{\pi} U_{IJ}$.
The fermion operators are given in terms of the bosons by $\psi_I^{\dagger} =  (2\pi a)^{-1/2} e^{\mp i \phi_I} \gamma_I$, where the sign is $-$ ($+$) for $I \leq N$ ($I > N$), $a$ is a short-distance cutoff,
and the Klein factors $\gamma_I$ satisfy $\gamma_I \gamma_J = - \gamma_J \gamma_I$ for $I \neq J$.

The LL action (\ref{eq:S_LL}) is a fixed point under RG flow,
parameterized by the symmetric positive definite $2N \times 2N$ matrix $V$.
Our results are based on a systematic linear stability analysis of these fixed points $S[V]$.
A generic perturbation of $S[V]$ has the form
\beq
\label{eqn:generic-pert}
S' = \int dt \, dx \, \Big[ \xi(x) \mathcal{O}(t,x) + \xi^*(x) \mathcal{O}^{\dagger}(t,x) \Big] ,
\eeq
where $\mathcal{O}$ is a local bosonic operator and $\xi(x)$ is an appropriate function.
It is natural to distinguish three types of perturbation: (i) \emph{global} perturbations, in which $\xi(x) = g e^{i\alpha}$ is constant in space, (ii) \emph{random} ones, in which $\xi(x)$ is a Gaussian random variable with $\overline{\xi(x)} = 0$ and $\overline{\xi^*(x) \xi(x')} = \sqrt{g} \, \delta(x-x')$, and (iii) \emph{local} ones, in which $\xi(x) = g e^{i\alpha} \delta(x-x_0)$ acts only at a point.
In each case, the linearized RG equation specifying how the coupling constant $g$ changes with the energy scale $\Lambda$ is
\beq
\label{eq:lin_rg}
\frac{d\ln g}{d\ln \Lambda} = \Delta - d_{\text{eff}},
\eeq
where $d_{\text{eff}} = 2,\frac{3}{2},1$ for global, random, or local perturbations respectively, and where $\Delta$ is the scaling dimension of $\mathcal{O}$.
The perturbation is relevant if $\Delta < d_{\text{eff}}$, marginal (at tree-level) if $\Delta = d_{\text{eff}}$, and irrelevant if $\Delta > d_{\text{eff}}$.

The quadratic action (\ref{eq:S_LL}) can be destabilized by localization or the opening of a gap, either of
which can be caused by a relevant perturbation (\ref{eqn:generic-pert})
if $\mathcal{O}$ is a
\emph{vertex operator} $\mathcal{O}_\mathbf{m} \equiv e^{i m_I \phi_I}$, where $\mathbf{m} \in \mathbb{Z}^{2N}$ (we suppress cutoff factors for brevity).
The operator $\mathcal{O}_\mathbf{m}$ is bosonic if and only if its \emph{conformal spin} $K(\mathbf{m}) = \tfrac{1}{2} \mathbf{m}^T \! K \mathbf{m}$ is an integer. 
At the fixed point $S[V]$, the scaling dimension of $\mathcal{O}_\mathbf{m}$ is
\beq
\Delta(\mathbf{m}) = \tfrac{1}{2} \mathbf{m}^T \! M \mathbf{m} ,
\eeq
where $M =  A^T\! A$, and $A \in SO(N,N)$ diagonalizes the interaction matrix, $A V \! A^T = \text{diag}(u_i)$.
Given any $M$, the set of corresponding interaction matrices can be parameterized as
\beq
\label{eq:V_param}
V = M^{-1/2} \begin{bmatrix} X & 0 \\ 0 & Y \end{bmatrix} M^{-1/2} ,
\eeq
where $X$ and $Y$ are arbitrary symmetric positive definite $N \times N$ matrices,
and $M^{-1/2}$ is the unique positive definite square root of $M^{-1}$ \cite{SuppMat}.
This parameterization of $V$ is closely related to, but distinct from, the one used in Refs.~\cite{Moore1998,Xu2006}.

We define two notions of stability of a LL fixed point $S[V]$.
We say that it is \emph{$\infty$-stable} if all non-chiral (i.e. $K(\mathbf{m}) = 0$) perturbations are irrelevant at $S[V]$.
We say that it is \emph{absolutely $\infty$-stable} if all chiral (i.e. $K(\mathbf{m}) \neq 0$) perturbations are irrelevant as well
\footnote{
Chiral perturbations cannot themselves lead to an energy gap, but one might worry that such perturbations, if relevant, will grow large enough to affect the scaling dimensions of non-chiral operators.
}.
The scaling dimensions are continuous functions of $V$, so each stable fixed point belongs to a \emph{stable phase}.
$\infty$-stable phases cannot exist when the LL has only $N=1$ channel. 
They can be shown to exist---by explicit construction---when $N \geq 23$ \cite{Plamadeala2014}.
In the intermediate range, $1 < N < 23$, the existence of $\infty$-stable phases remains an open question at this time.
Meanwhile, upper bounds on the density of high-dimensional sphere packings \cite{Cohn2003} imply that absolutely $\infty$-stable phases cannot exist with $N < 11$ channels. 
It is again possible to show---by explicit construction---that they do exist when $N$ is sufficiently large.
For completeness, we discuss these matters in more detail in the Supplemental Material~\cite{SuppMat}.

From a physical point of view, however, the notions of stability introduced above are unnecessarily restrictive.
If there are physical reasons to expect the bare value $g_0$ of a relevant coupling to be small, 
then although this coupling will eventually destabilize the metallic state, this will only happen at very low temperatures
$T \sim \Lambda_0 \, g_0^{1/(d_{\text{eff}} - \Delta)}$.
We expect $g_0$ to be small for perturbations that are sufficiently high-order in the fermion fields. 
This is based on the assumption that such terms are not appreciably generated during RG flow from the underlying microscopic theory (which only has terms up to quartic order) to the effective theory $S_{\text{eff}}$ which describes the system at energies $\sim \Lambda_0$.

Each vertex operator $\mathcal{O}_\mathbf{m}$ in the bosonic formulation corresponds to terms that are $\abs{\mathbf{m}}^{\text{th}}$-order in the fermion fields, where $\abs{\mathbf{m}} \equiv \sum_{I=1}^{2N} \abs{m_I}$.
We say that the fixed point $S[V]$ is \emph{$q$-stable} if $q$ is the largest integer such that all non-chiral perturbations of $S[V]$ with $\abs{\mathbf{m}} \leq q$ are irrelevant.
We say that it is \emph{absolutely $q$-stable} if $q$ is the largest integer such that \emph{all} perturbations with $\abs{\mathbf{m}} \leq q$ are irrelevant.
Our earlier notions of stability are the limiting cases $q = \infty$.
Based on the comments in the previous paragraph, it is plausible that, in any real system, there will
be no observable difference between $q$-stability 
and $\infty$-stability 
at accessible temperatures if $q$ is sufficiently large
\footnote{Even if this assumption turns out to be false, a $q$-stable phase can be expected to exhibit novel and exotic instabilities, since all the usual instabilities correspond to operators with small $\abs{\mathbf{m}}$.}.

\vspace{0.5em}
\textbf{Relation to Integral Quadratic Lattices.}
As described in the Introduction, there is a beautiful geometric picture associated with all of this.
To any interaction matrix $V$ diagonalized by $A \in SO(N,N)$, we associate a lattice $A \mathbb{Z}^{2N} \equiv \{ A \mathbf{m} \mid \mathbf{m} \in \mathbb{Z}^{2N} \} $ in a fictitious $\mathbb{R}^{2N}$ equipped with two metrics: 
the mixed-signature $(N,N)$ metric $K = \mathrm{diag}(-\mathbb{I}_N , \mathbb{I}_N)$ and the Euclidean metric $\mathbb{I}_{2N}$.
The scaling dimension of an operator is equal to half the Euclidean interval from the origin to the associated lattice point, $\Delta(\mathbf{m}) = \frac{1}{2} \norm{A \mathbf{m}}^2$.
There are three ``spheres of relevance'' centered at the origin, with Euclidean radii $\sqrt{2d_{\text{eff}}} = 2, \sqrt{3}, \sqrt{2}$;
any lattice point inside these spheres represents a perturbation that is relevant if it is global, random, or local, respectively.
The chirality (i.e.~conformal spin) of an operator is equal to half the mixed-signature interval from the origin to the associated lattice point;
chiral operators correspond to ``spacelike'' or ``timelike'' intervals, and non-chiral operators to ``lightlike'' (null) intervals.
The lattice is naturally graded into ``shells'' of fixed $\abs{\mathbf{m}} \equiv \sum_{I=1}^{2N} \abs{m_I}$, 
which equals the order of the corresponding perturbation $\mathcal{O}_\mathbf{m}$ in the fermion fields.
Bosonic operators have even $\abs{\mathbf{m}}$.
The fixed point $S[V]$ is $q$-stable if no lightlike even lattice point in the innermost $q$ shells falls within the sphere of Euclidean radius $2$ centered at the origin.
It is absolutely $q$-stable if the same also holds for spacelike and timelike even lattice points in these shells.

\begin{figure}[t]
    \centering
    \includegraphics[width=0.8\columnwidth]{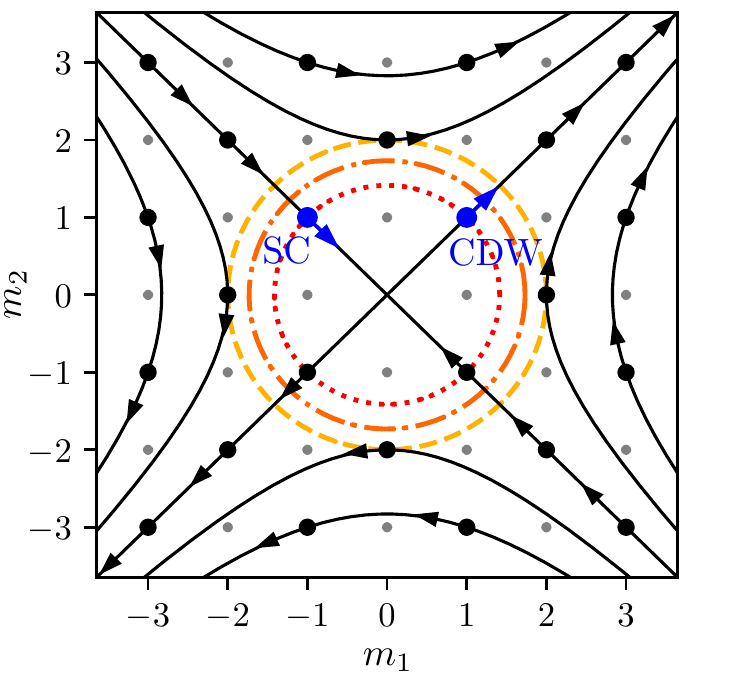}
    \caption{Lattice of perturbations for an $N=1$ channel LL.
    Large dots are bosonic operators, while small gray dots are fermionic ones; the latter can be ignored.
    A perturbation is relevant if it falls within the appropriate circle
    ({\color[rgb]{1,0.7,0} \hdashrule[0.4ex][x]{20pt}{1.2pt}{4pt 1.5pt}}~global,
    {\color[rgb]{1,0.4,0} \hdashrule[0.4ex][x]{22pt}{1.2pt}{6pt 1.5pt 1.2pt 1.5pt}}~random,
    {\color[rgb]{1,0,0} \hdashrule[0.4ex][x]{19pt}{1.2pt}{1.2pt 1.5pt}}~local).
    The lattice shown is $\mathbb{Z}^2$, corresponding to the noninteracting fixed point, $\lambda=0$.
    With attractive interactions, $\lambda < 0$, the lattice deforms as indicated by the flow field. 
    With repulsive interactions, $\lambda > 0$, the flow is in the opposite direction.
    }
    \label{fig:1channel_boosts}
\end{figure}

Figure~\ref{fig:1channel_boosts} illustrates these ideas in the simplest case, that of $N=1$ channel.
The matrix $A \in SO(1,1)$ then describes a boost (hyperbolic rotation) of the plane, and can be parameterized as
$A(\lambda) = e^{-(\lambda/2) \sigma_x}$.
At the noninteracting fixed point, $\lambda = 0$, the most relevant perturbations couple $\mathcal{O}_{\text{SC}} \equiv \psi^{\dagger}_R \psi^{\dagger}_L \sim e^{i (-\phi_1 + \phi_2)}$ to an external 3D superconductor, or 
$\mathcal{O}_{\text{CDW}} \equiv \psi^{\phantom{\dagger}}_R \psi^{\dagger}_L \sim e^{i (\phi_1 + \phi_2)}$ to a periodic potential.
The corresponding lattice points are $\mathbf{m} = (-1,1)$ and $\mathbf{m} = (1,1)$ respectively.
When $\lambda = 0$, both operators have $\Delta = 1$, so both perturbations are relevant; the associated instabilities are induced superconductivity (SC) and a pinned charge density wave (CDW) respectively.
When interactions are turned on, so that $\lambda \neq 0$, the lattice deforms to $A(\lambda) \mathbb{Z}^2$ as indicated in the Figure.
Thus, $\lambda < 0$ makes $\mathcal{O}_{\text{CDW}}$ less relevant but $\mathcal{O}_{\text{SC}}$ more relevant, while $\lambda > 0$ does the opposite.
The interaction matrix $V$ can be parametrized as in Eq.~(\ref{eq:V_param}), with $X = u_1 > 0$, $Y = u_2 > 0$, and $M^{1/2} = A(-\lambda) = e^{(\lambda/2) \sigma_x}$;
its off-diagonal element is $V_{12} = \frac{1}{2}(u_1+u_2) \sinh{\lambda}$.
Thus, $\lambda < 0$ ($\lambda > 0$) corresponds to attractive (repulsive) interactions, and we reproduce the well-known phenomenology of the 1-channel Luttinger liquid \cite{Giamarchi2003}.
Clearly, stability is impossible with just $N=1$ channel.

\vspace{0.5em}
\textbf{Stable Luttinger Liquids.}
We now turn to the general case of $N$ channels.
Our approach is to study all possible scaling dimension matrices $M$. After we have identified some $M$'s of interest, we reconstruct the corresponding $V$'s using Eq.~(\ref{eq:V_param}).

A useful structure theorem for $SO(N,N)$, called the \emph{hyperbolic cosine-sine (CS) decomposition} \cite{Higham2003}, ensures that $M$ can be written as a product of independent boosts in orthogonal planes:
\beq
\label{eq:hyp_CS}
M = \begin{bmatrix}
Q_1^T & 0 \\ 0 & Q_2^T
\end{bmatrix}
\begin{bmatrix}
C & -S \\ -S & C
\end{bmatrix}
\begin{bmatrix}
Q_1 & 0 \\ 0 & Q_2
\end{bmatrix} ,
\eeq
where $Q_1,Q_2 \in SO(N)$, $C = \text{diag}(\cosh{\lambda_i})$, and $S = \text{diag}(\sinh{\lambda_i})$, with $\lambda_i \in \mathbb{R}$, $i = 1,2,\dots,N$.

The \emph{crucial geometric fact} distinguishing $N \geq 2$ from $N=1$ is that the boost planes of $M$ can be rotated out of alignment with the lattice planes of $\mathbb{Z}^{2N}$ by suitably chosen $Q_i$.
As a consequence, for $N \geq 2$, absolutely $q$-stable phases exist for any finite $q$.
This assertion can be proven quite simply, as follows.

Take $\lambda_i = \lambda$ in expression (\ref{eq:hyp_CS}) for $M$, and let $\mathbf{m} = (\mathbf{m}_R, \mathbf{m}_L)$, with $\mathbf{m}_{R/L} \in \mathbb{Z}^N$.
If either $\mathbf{m}_R$ or $\mathbf{m}_L$ vanishes, then $\Delta(\mathbf{m}) = \frac{1}{2} \norm{\mathbf{m}}^2 \cosh{\lambda} > 2$ for $\lambda > \mathrm{arccosh} \, 2$.
If neither $\mathbf{m}_R$ nor $\mathbf{m}_L$ vanishes, we can rewrite the inequality $\Delta(\mathbf{m}) > 2$ as
\beq
\label{eq:stable_ineq}
\abs{\hat{\mathbf{m}}_R^T Q \hat{\mathbf{m}}_L}  < f\Big(\frac{\norm{\mathbf{m}_R}}{\norm{\mathbf{m}_L}}\Big) \coth{\lambda} - \frac{2 \csch{\lambda}}{\norm{\mathbf{m}_R} \norm{\mathbf{m}_L}} ,
\eeq
where $Q \equiv Q_1^T Q_2 \in SO(N)$, $\mathbf{\hat{m}}_{\nu} \equiv \mathbf{m}_{\nu} / \norm{\mathbf{m}_{\nu}}$ and $f(x) \equiv \frac{1}{2}(x+x^{-1})$.
There are a finite number of vectors $\mathbf{m} \in \mathbb{Z}^{2N}$ that satisfy $\abs{\mathbf{m}} \leq q$, so the unit vectors $\mathbf{\hat{m}}_{R/L}$ in Eq.~(\ref{eq:stable_ineq}) belong to a finite set $\Omega_q$.
This set cannot fill the unit sphere densely, so there exists $Q \in SO(N)$ and $\epsilon > 0$ such that $\abs*{\mathbf{\hat{m}}_R^T Q \hspace{1pt} \mathbf{\hat{m}}_L} < 1 - \epsilon$ for all $\mathbf{\hat{m}}_{R/L} \in \Omega_q$.
But $f(x) \coth{\lambda} > 1$ for any $x, \lambda > 0$, while $\csch{\lambda} \to 0$ as $\lambda \to \infty$.
Thus the right side of Eq.~(\ref{eq:stable_ineq}) is greater than $1 - \epsilon$ for sufficiently large $\lambda$. \qed

In the $N = 2$ channel case, $M$ is parameterized, according to Eq.~(\ref{eq:hyp_CS}), by two rapidities ($\lambda_1, \lambda_2$) and two angles ($\theta_1, \theta_2$, where $\theta_i$ is the rotation angle of $Q_i \in SO(2)$).
It is convenient to write these as
\beq
\lambda_{1,2} = \delta \pm \lambda , \qquad
\theta_{1,2} = \tfrac{1}{2} (\theta \mp \alpha) .
\eeq
In the limit $\delta \to 0$, the dependence on $\alpha$ disappears.
The full parameterization of $M$ is written down explicitly in the Supplemental Material~\cite{SuppMat}.

\begin{figure}[t]
    \centering
    \includegraphics[width=0.95\columnwidth]{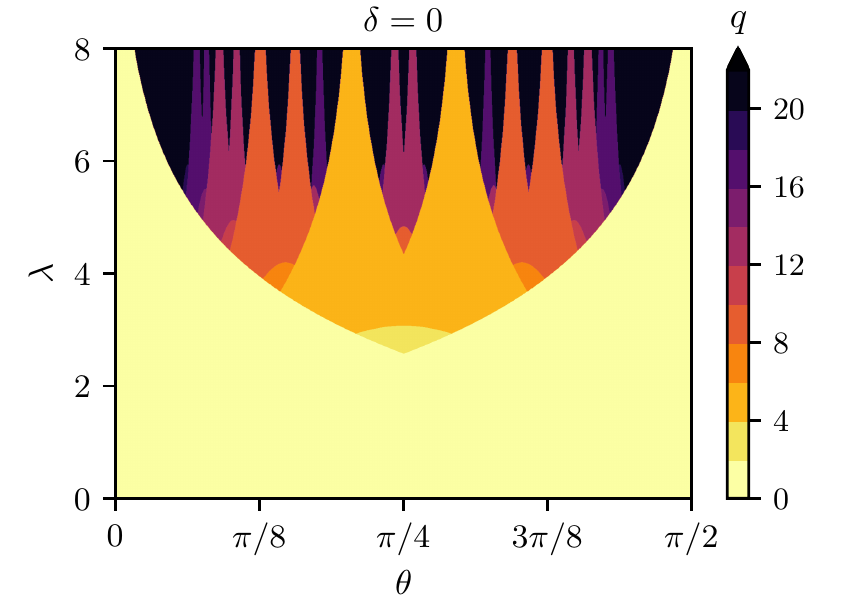}
    \caption{A slice of the absolute $q$-stability phase diagram for the $N=2$ channel LL.
    Each point on the plot is assigned the largest integer $q$ such that all perturbations with $\abs{\mathbf{m}} \leq q$ are irrelevant for those parameter values $(\lambda,\theta)$.
    The diagram is identical for $\theta \mapsto \theta + n\pi/2$.}
    \label{fig:absolute_q_stability_d0_a0}
\end{figure}

We construct an ``absolute $q$-stability phase diagram'' for the $2$-channel LL by assigning to each point $(\lambda,\delta,\theta,\alpha)$ in the resulting parameter space its absolute $q$-stability value, $q$.
Figure~\ref{fig:absolute_q_stability_d0_a0} shows the $\delta = 0$ slice of this diagram; other slices may be found in Ref.~\cite{SuppMat}.
Each point in the phase diagram corresponds to a 6-parameter family of interaction matrices $V$, which can be obtained using Eq.~(\ref{eq:V_param}).
The resulting general expression for $V$ is given in Ref.~\cite{SuppMat}.
Here, we concentrate on the particular case in which the diagonal blocks $V_{RR}$ and $V_{LL}$ are equal, and $\delta = 0$. 
In this case,
\beq
\label{eq:V2}
V = \left[ \begin{array}{cc|cc}
v_+ & w & c_+ & c_0 \\
w & v_- & c_0 & c_- \\[0.15em]
\hline
c_+ & c_0 & v_+ & w \\[-0.15em]
c_0 & c_- & w & v_-
\end{array} \right] ,
\eeq
where $v_{\pm} = v \pm u$,
\begin{subequations}
\begin{align}
c_{\pm} &= (w \sin{\theta} \pm v_{\pm} \cos{\theta}) \tanh{\lambda} , \\*
c_0 &= v \sin{\theta} \tanh{\lambda} .
\end{align}
\end{subequations}
The parameters $v$, $u$, and $w$ do not affect scaling dimensions; they can be chosen arbitrarily subject only to the constraint that $V$ must be positive definite, which requires $v > 0$ and
\beq
\label{eq:pos_def}
(u \sin{\theta} - w \cos{\theta})^2 \cosh^2\!\lambda + (u \cos{\theta} + w \sin{\theta})^2 < v^2 .
\eeq
If in addition $\theta \in [0,\pi]$ and $u \cos{\theta} + w \sin{\theta} \geq v \abs{\cos{\theta}}$, then every entry in the $V$ matrix is nonnegative.
Note that the above inequalities can be satisfied simultaneously---the first defines the interior of an ellipse in the $(u/v, w/v)$ plane, and the second selects a segment of this ellipse.
Thus, we can realize any of the absolutely $q$-stable phases in Figure~\ref{fig:absolute_q_stability_d0_a0} with purely repulsive interactions.

The 2-channel LLs defined by Eqs.~(\ref{eq:S_LL}) and (\ref{eq:V2}--\ref{eq:pos_def}) can in principle be realized in a single-spinful-channel quantum wire with either time-reversal or spatial inversion symmetry, but not both \cite{SuppMat}. 
(We emphasize that these LL phases are $q$-stable with respect to perturbations that break the symmetry as well.)
If the system also has spin-rotation symmetry about some axis, one can reformulate the effective action in terms of non-chiral charge and spin fields \cite{Giamarchi2003,SuppMat}.
In the time-reversal invariant case, the corresponding Hamiltonian takes the form
\begin{align}
\label{eq:H2}
H = \frac{1}{2\pi} \int dx \, \bigg[ 
&v_c K_c (\pi \Pi_c)^2 + \frac{v_c}{K_c} (\partial_x \varphi_c)^2 \nonumber \\*
&+ v_s K_s (\pi \Pi_s)^2 + \frac{v_s}{K_s} (\partial_x \varphi_s)^2 \nonumber \\*
&+ d_+ (\pi \Pi_c) \partial_x \varphi_s + d_- (\pi \Pi_s) \partial_x \varphi_c \bigg] ,
\end{align}
where $\varphi_{c(s)}$ is the charge (spin) phase field, with conjugate momentum density $\Pi_{c(s)}$.
The parameters $v_{c(s)}, K_{c(s)}, d_{\pm}$ are simple functions of $v_{\pm}, w, c_{\pm}, c_0$; the expressions are given in Ref.~\cite{SuppMat}.
The $q$-stable phases identified in this work require $d_+ \neq d_-$.
Standard treatments of a spin-orbit-coupled LL, such as Ref.~\cite{Moroz2000}, make the additional assumption that interactions are pointlike; this leads to Eq.~(\ref{eq:H2}) with $d_+ = d_-$.
However, $d_+ \neq d_-$ is perfectly consistent with the symmetries of the problem, and appears naturally if one allows for more general short-range interactions.

\vspace{0.5em}
\textbf{Conclusions.}
As we have seen in this paper, the 1-channel Luttinger liquid is the exception.
For any number of channels $N>1$---including even $N=2$---there are parameter regimes in which, for any desired finite $q$,
all instabilities up to $q$-th order in electron operators are kept at bay.
These phases are, in some sense, better examples of non-Fermi liquids than the 1-channel LL
since they do not have a tendency to order frustrated only by low dimension.
We cannot take $q=\infty$, so these states will eventually be unstable, but this may not occur until unobservably low temperatures. 
Moreover, it is much more likely that it will be possible to tune the parameters of an $N=2$ channel Luttinger liquid into the necessary regime in an experiment than it would be for $N=23$, which appears to be necessary for $q=\infty$. 
Thus, the work in this paper may facilitate the observation of these phases in experiments and may serve as a useful paradigm for
thinking about higher-dimensional non-Fermi liquids. 
Our results can also be translated into statements about stable phases of classical 2-dimensional or layered 3-dimensional systems; it would be interesting to explore the consequences for particular classical systems of experimental interest.

\begin{acknowledgments}

We thank Michael Freedman and Eugeniu Plamadeala for helpful discussions, and Eduardo Fradkin, Steven A. Kivelson and Michael Mulligan for useful comments on the manuscript.
This work was supported by the Microsoft Corporation.

\end{acknowledgments}

% Bibliography
% =======================================================================


% Supplemental material (title)
% =======================================================================

\onecolumngrid
\clearpage

\begin{center}

{\large \bf Supplemental Material for\\[0.2em]
``Almost Perfect Metals in One Dimension''}\\[0.8em]

Chaitanya Murthy${}^1$ and Chetan Nayak${}^{1,2}$\\[0.5em]
{\small ${}^1$\it Department of Physics, University of California, Santa Barbara, CA 93106, USA\\[-0.05em]
${}^2$Microsoft Quantum, Station Q, University of California, Santa Barbara, CA 93106, USA}

\end{center}

\vspace{0.3cm}

\twocolumngrid

\setcounter{figure}{0}
\setcounter{equation}{0}

\renewcommand{\thefigure}{S\arabic{figure}}
\renewcommand{\theequation}{S\arabic{equation}}
\renewcommand{\thesection}{S\arabic{section}}

\newcommand{\scite}[1]{[\textcolor{blue!70!black}{S}\citenum{#1}]}

% Supplemental material (body)
% =======================================================================

\section{Restrictions on the interaction matrix $V$ imposed by 
%time-reversal ($\mathcal{T}$) and/or inversion ($\mathcal{P}$) 
symmetries}
\label{app:V_sym_restrictions}

Consider the effective theory of a system that is invariant under one or more symmetries that interchange right and left-movers, such as time-reversal ($\mathcal{T}$), and/or spatial inversion ($\mathcal{P}$).
On general grounds, $\mathcal{T}$ must be implemented in the effective theory by an anti-unitary operator that squares to $-1$ when acting on fermionic operators.
Spatial inversion $\mathcal{P}$ must be implemented by a unitary operator that squares to $+1$.

\subsection{$\mathcal{T}$ symmetry but no $\mathcal{P}$ symmetry}
\label{app:V_sym_T}

First consider the case in which the system has time-reversal symmetry but no inversion symmetry.
The chiral boson fields $\phi_I$ can be chosen to transform as follows under time-reversal (here the index $I = 1,2, \dots N$):
\beq
\label{eq:TR_def}
\mathcal{T} :
\begin{cases}
\phi_I(x,t) &\longrightarrow \quad \phi_{I+N}(x,-t) , \\
\phi_{I+N}(x,t) &\longrightarrow \quad \pi + \phi_I(x,-t) .
\end{cases}
\eeq
In addition, $\mathcal{T}$ complex conjugates $i \to -i$.
Then, $\mathcal{T}$ correctly interchanges right- and left-movers, and squares to $-1$ when acting on the fermion fields $\psi_I \propto e^{\mp i \phi_I} \gamma_I$.
(Alternatively, one could omit the $\pi$ in Eq.~(\ref{eq:TR_def}) and have the Klein factors $\gamma_I$ transform nontrivially.)
In this representation, $\mathcal{T}$ symmetry imposes that the interaction matrix $V$ must satisfy
\beq
\label{eq:V_sym1_cond}
\Sigma V \Sigma = V ,
\eeq
where $\Sigma \equiv \sigma_x \otimes \mathbb{I}_N$ and $\sigma_x$ is the usual Pauli matrix.
Thus, $V$ must have the block form
\beq
\label{eq:V_sym1_form}
V = \begin{bmatrix}
V_1 & V_2 \\ V_2 & V_1
\end{bmatrix} ,
\eeq
where $V_i = V_i^T$. Conversely, any $2N \times 2N$ positive definite matrix $V$ of this form can serve as the interaction matrix of a $\mathcal{T}$-symmetric $N$-channel Luttinger liquid.

\subsection{$\mathcal{P}$ symmetry but no $\mathcal{T}$ symmetry}
\label{app:V_sym_P}

Next consider the case in which the system has inversion symmetry but no time-reversal symmetry.
The chiral boson fields $\phi_I$ can be chosen to transform as follows under spatial inversion (again the index $I = 1,2, \dots N$):
\beq
\label{eq:inv_def}
\mathcal{P} :
\begin{cases}
\phi_I(x,t) &\longrightarrow \quad - \phi_{I+N}(-x,t) , \\
\phi_{I+N}(x,t) &\longrightarrow \quad - \phi_I(-x,t) .
\end{cases}
\eeq
Then, $\mathcal{P}$ correctly interchanges right- and left-movers, and squares to $+1$ when acting on the fermion fields $\psi_I$.
In this representation, $\mathcal{P}$ symmetry imposes that the interaction matrix $V$ must satisfy Eq.~(\ref{eq:V_sym1_cond}), and hence that it must have the block form (\ref{eq:V_sym1_form}).
Conversely, any $2N \times 2N$ positive definite matrix $V$ of the form (\ref{eq:V_sym1_form}) can serve as the interaction matrix of a $\mathcal{P}$-symmetric $N$-channel Luttinger liquid.

\subsection{Both $\mathcal{T}$ and $\mathcal{P}$ symmetry}
\label{app:V_sym_PT}

Finally, consider the case in which the system has both time-reversal symmetry and inversion symmetry.
The transformation laws (\ref{eq:TR_def}) and (\ref{eq:inv_def}) correspond to \emph{different} representations of the fermion fields in terms of bosons, and hence cannot be used simultaneously.
As is well known, symmetry with respect to $\mathcal{P} \mathcal{T}$ enforces a twofold degeneracy of the bands at each point in $k$-space. Therefore, the low-energy effective theory now involves $4N$ chiral spinless Dirac fermions $\psi_I$, where $I = 1,2, \dots,2N$ labels right-movers and $I = 2N+1,\dots,4N$ labels left-movers.
The corresponding chiral boson fields $\phi_I$ can be chosen to transform as follows under time-reversal and spatial inversion (here the index $I = 1,2, \dots N$):
\begin{subequations}
\label{eq:TR_inv_def}
\begin{align}
\mathcal{T} :
&\begin{cases}
\phi_I(x,t) &\longrightarrow \quad \phi_{I+2N}(x,-t) , \\
\phi_{I+N}(x,t) &\longrightarrow \quad \phi_{I+3N}(x,-t) , \\
\phi_{I+2N}(x,t) &\longrightarrow \quad \pi + \phi_I(x,-t) , \\
\phi_{I+3N}(x,t) &\longrightarrow \quad \pi + \phi_{I+N}(x,-t) ,
\end{cases} \\*
\mathcal{P} :
&\begin{cases}
\phi_I(x,t) &\longrightarrow \quad - \phi_{I+3N}(-x,t) , \\
\phi_{I+N}(x,t) &\longrightarrow \quad - \phi_{I+2N}(-x,t) , \\
\phi_{I+2N}(x,t) &\longrightarrow \quad - \phi_{I+N}(-x,t) , \\
\phi_{I+3N}(x,t) &\longrightarrow \quad - \phi_I(-x,t) .
\end{cases}
\end{align}
\end{subequations}
Now, $\mathcal{T}$ symmetry and $\mathcal{P}$ symmetry respectively impose that the interaction matrix $V$ must satisfy
\begin{subequations}
\label{eq:V_sym2_cond}
\begin{align}
(\mathcal{T}) \qquad \Sigma_1 V \Sigma_1 &= V , \\*
(\mathcal{P}) \qquad \Sigma_2 V \Sigma_2 &= V ,
\end{align}
\end{subequations}
where $\Sigma_1 \equiv \sigma_x \otimes \mathbb{I}_2 \otimes \mathbb{I}_N$ and
$\Sigma_2 \equiv \sigma_x \otimes \sigma_x \otimes \mathbb{I}_N$.
Thus, $V$ must have the block form
\beq
\label{eq:V_sym2_form}
V = \begin{bmatrix}
V_1 & V_2 & V_3 & V_4 \\ 
V_2 & V_1 & V_4 & V_3 \\ 
V_3 & V_4 & V_1 & V_2 \\ 
V_4 & V_3 & V_2 & V_1
\end{bmatrix} ,
\eeq
where $V_i = V_i^T$.
Conversely, any $4N \times 4N$ positive definite matrix $V$ of this form can serve as the interaction matrix of a $\mathcal{T}$- and $\mathcal{P}$-symmetric $2N$-channel Luttinger liquid.

\section{Properties of the map $\varphi : V \mapsto M$ from interaction matrices to scaling dimension matrices}
\label{app:V_to_M_properties}

Let $\mathscr{P}_N$ denote the set of real symmetric positive definite $N \times N$ matrices, and let $\mathscr{M}_N \equiv SO(N,N) \cap \mathscr{P}_{2N}$.
The map $\varphi$ from interaction matrices $V \in \mathscr{P}_{2N}$ to ``scaling dimension matrices'' $M \in \mathscr{M}_N$ is defined as
\beq
\begin{aligned}
\varphi &: \mathscr{P}_{2N} \to \mathscr{M}_N , \\
\varphi &: V \mapsto A^T \! A ,
\end{aligned}
\eeq
where $A \in SO(N,N)$ and $A V \! A^T = D$ is diagonal.

\subsection{General properties}

The first and second lemmas below show that $\varphi$ is well-defined.
The third, fourth and fifth lemmas characterize the inverse images $\varphi^{-1}(M)$, and yield the parameterization of $V$ matrices used in the main text.
All of these results are elementary, but we record them here for completeness.

\begin{lem}
If $V \in \mathscr{P}_{2N}$, then there exists $A \in SO(N,N)$ such that $A V \! A^T = D$ is diagonal.
\end{lem}
\begin{proof}
(by construction).
Let $V^{1/2}$ denote the unique symmetric positive definite square root of $V$, so that
\beq
V^{1/2} = (V^{1/2})^T , \qquad 
V^{1/2} > 0, \qquad
(V^{1/2})^2 = V ,
\eeq 
and let $V^{-1/2} \equiv (V^{1/2})^{-1} = (V^{-1})^{1/2}$.
The matrix $V^{-1/2} K V^{-1/2}$ (where $K = -\mathbb{I}_N \oplus \mathbb{I}_N$) is symmetric, and can therefore be diagonalized by some $Q \in SO(2N)$.
Furthermore, Sylvester's theorem of inertia \scite{Sylvester1852} ensures that $V^{-1/2} K V^{-1/2}$ has $N$ positive and $N$ negative eigenvalues.
Thus, $Q$ can be chosen so that
\beq
\label{eq:VKV_diag}
Q V^{-1/2} K V^{-1/2} Q^T = D^{-1} K ,
\eeq
where $D$ is diagonal and positive definite (this can be arranged by re-ordering the rows of $Q$ and, if necessary, multiplying one row by $-1$ to maintain $\det Q = +1$).
Taking the determinant of both sides of Eq.~(\ref{eq:VKV_diag}), we have $\det V = \det D$.
Therefore $A \equiv D^{1/2} Q V^{-1/2}$ satisfies the desired properties: $A K A^T = K$, $\det A = 1$, and $A V \! A^T = D$.
\end{proof}

\begin{lem}
If $A_i \in SO(N,N)$  and $A_i V \! A_i^T = D_i$ is diagonal for $i = 1,2$, then $A_1^T A_1 = A_2^T A_2$.
\end{lem}
\begin{proof}
Note that every $A \in SO(N,N)$ is invertible, with $A^{-1} = K A^T K$, where $K = -\mathbb{I}_N \oplus \mathbb{I}_N$ (this follows immediately from the defining condition for the group, $A K A^T = K$, and the fact that $K^2 = \mathbb{I}_{2N}$.).
Thus, to prove the lemma it suffices to prove the equivalent statement that $D_2 = A D_1 A^T$ implies $A^T \! A = \mathbb{I}_{2N}$, where $A \equiv A_2 A_1^{-1} \in SO(N,N)$.

Using $A^T = K A^{-1} K$, the equation $D_2 = A D_1 A^T$ can be rewritten as
\beq
\label{eq:DK_similarity}
D_2 K = A (D_1 K) A^{-1} .
\eeq
Thus, the matrices $D_2 K$ and $D_1 K$ are similar.
But similar diagonal matrices can differ only by a permutation of the diagonal elements.
Taking account of the sign structure due to $K$, one must have $D_2 = P D_1 P^{-1}$, with $P = P^{(1)} \oplus P^{(2)}$, where the $P^{(i)}$ are $N \times N$ permutation matrices.
Defining $B \equiv P^{-1} \! A$, Eq.~(\ref{eq:DK_similarity}) reduces to
\beq
D_1 K = B (D_1 K) B^{-1} .
\eeq
This implies that $B$ preserves each eigenspace of $D_1 K$. 
Hence it must (at the very least) have the block form $B = B^{(1)} \oplus B^{(2)}$, where the $B^{(i)}$ are $N \times N$ matrices.
Since $A = PB$, and $P$ has a similar block structure, one must also have $A = A^{(1)} \oplus A^{(2)}$, where the $A^{(i)}$ are $N \times N$ matrices.
Then the condition $A K A^T = K$ implies $A^{(i)} \in O(N)$, so that $A^T \! A = \mathbb{I}_{2N}$.
\end{proof}

\begin{lem}\label{lem:V(M)}
$V \in \varphi^{-1}(M)$ if and only if
\beq
\label{eq:V_param}
V = M^{-1/2} 
\begin{bmatrix}
X & 0 \\ 0 & Y
\end{bmatrix}
M^{-1/2}
\eeq
for some $X,Y \in \mathscr{P}_N$, where $M^{-1/2}$ denotes the unique positive definite square root of $M^{-1}$.
\end{lem}
\begin{proof}
$(\Longrightarrow)$:
Assume that $\varphi(V) = M$.
Every $M \in \mathscr{M}_N = SO(N,N) \cap \mathscr{P}_{2N}$ has a unique positive definite symmetric square root $M^{1/2} \in \mathscr{M}_N$.
Furthermore, any matrix $A$ that satisfies $A^T \! A = M$ can be written as $A = R M^{1/2}$, for a suitable $R \in O(2N)$.
If $A \in SO(N,N)$, then we must have $R \in O(2N) \cap SO(N,N) = O(N) \times O(N) / \mathbb{Z}_2$.
Therefore, $\varphi(V) = M$ implies that $(R M^{1/2}) V (R M^{1/2})^T = D$ for some diagonal positive definite $D$ and some $R \in O(N) \times O(N) / \mathbb{Z}_2$; equivalently, $V = M^{-1/2} R^T D R M^{-1/2}$, which is of the form indicated.

$(\Longleftarrow)$:
Assume that $V$ has the form indicated. Then there exist $R_1, R_2 \in SO(N)$ that diagonalize $X,Y$ respectively.
Let $A \equiv [R_1 \oplus R_2] M^{1/2}$. Then $A \in SO(N,N)$, $A V \! A^T = D$ is diagonal, and $A^T \! A = M$. 
Thus $\varphi(V) = M$.
\end{proof}

The scaling dimension matrix $M \in \mathscr{M}_N$ can, by the hyperbolic CS decomposition, be written as [Eq.~(7) of the main text]:
\beq
\label{eq:app_hyp_CS}
M = \begin{bmatrix}
Q_1^T & 0 \\ 0 & Q_2^T
\end{bmatrix}
\begin{bmatrix}
C & -S \, \\ -S & C
\end{bmatrix}
\begin{bmatrix}
Q_1 & 0 \\ 0 & Q_2
\end{bmatrix} ,
\eeq
where $Q_1,Q_2 \in O(N)$, $C = \text{diag}(\cosh{\lambda_i})$, and $S = \text{diag}(\sinh{\lambda_i})$, with $\lambda_i \geq 0$, $i = 1,2,\dots,N$.
Note that we can equivalently take $Q_1, Q_2 \in SO(N)$ if we allow one of the $\lambda_i$'s to be negative, as done in the main text.

\begin{lem}\label{lem:V(M)_alt}
$V \in \varphi^{-1}(M)$ if and only if
\beq
\label{eq:V_param_alt}
V = \begin{bmatrix}
X + F \, Y F^T & X F + F \, Y \\ F^T \! X + Y F^T & F^T \! X F + Y
\end{bmatrix} ,
\eeq
for some $X,Y \in \mathscr{P}_N$, where
\beq
\label{eq:F_def}
F \equiv Q_1^T \,\text{diag}(\tanh(\lambda_i/2)) \, Q_2 .
\eeq
\end{lem}
\begin{proof}
According to Lemma~\ref{lem:V(M)}, $V \in \varphi^{-1}(M)$ iff $V = M^{-1/2} [\tilde{X} \oplus \tilde{Y}] M^{-1/2}$ for some $\tilde{X},\tilde{Y} \in \mathscr{P}_N$.
From Eq.~(\ref{eq:app_hyp_CS}), it follows that
\beq
\label{eq:M_sqrt_inv}
M^{-1/2} = \begin{bmatrix}
\tilde{C}_1 & \tilde{S} \\ \tilde{S}^T & \tilde{C}_2
\end{bmatrix} ,
\eeq 
where $\tilde{C}_{\nu} = Q_{\nu}^T \, \text{diag}(\cosh(\lambda_i/2)) \, Q_{\nu}$ ($\nu = 1,2$) and $\tilde{S} = Q_1^T \, \text{diag}(\sinh(\lambda_i/2)) \, Q_2$.
Thus,
\beq
V = \begin{bmatrix}
\tilde{C}_1 \tilde{X} \tilde{C}_1 + \tilde{S} \, \tilde{Y} \tilde{S}^T &  \tilde{C}_1 \tilde{X} \tilde{S} + \tilde{S} \, \tilde{Y} \tilde{C}_2 \\ 
\tilde{S}^T \! \tilde{X} \tilde{C}_1 + \tilde{C}_2 \tilde{Y} \tilde{S}^T & \tilde{S}^T \! \tilde{X} \tilde{S} + \tilde{C}_2 \tilde{Y} \tilde{C}_2
\end{bmatrix} .
\eeq
Now define $X \equiv \tilde{C}_1 \tilde{X} \tilde{C}_1$ and $Y \equiv \tilde{C}_2 \tilde{Y} \tilde{C}_2$.
These maps from $\tilde{X},\tilde{Y} \in \mathscr{P}_N$ to $X,Y \in \mathscr{P}_N$ are bijections, because $\tilde{C}_{\nu} \in \mathscr{P}_N$.
Noting that $\tilde{S} \tilde{C}_2^{-1} = \tilde{C}_1^{-1} \tilde{S} = F$, we obtain the claimed result, Eq.~(\ref{eq:V_param_alt}).
\end{proof}

We now write the interaction matrix in block form as
\beq
V = \begin{bmatrix}
V_{RR} & V_{RL} \\ V_{LR} & V_{LL}
\end{bmatrix} ,
\eeq
where $V_{RR} , V_{LL} \in \mathscr{P}_N$ and $V_{LR} = V_{RL}^T$.

\begin{lem}
\label{lem:V(M)_imp}
$V \in \varphi^{-1}(M)$ if and only if
\begin{subequations}
\begin{align}
V_{RR} &\in \mathscr{P}_N , \label{eq:V_pos_1}\\*
V_{RR} - V_{RL} \, V_{LL}^{-1} \, V_{RL}^T &\in \mathscr{P}_N , \label{eq:V_pos_2} \\*
V_{RL} + F \, V_{RL}^T F - V_{RR} F - F V_{LL} &= 0 , \label{eq:V_M_cond}
\end{align}
\end{subequations}
where $F$ is defined in Eq.~(\ref{eq:F_def}) above.
\end{lem}
\begin{proof}
Equations~(\ref{eq:V_pos_1}) and (\ref{eq:V_pos_2}) are the Schur complement condition for positive definiteness of a symmetric matrix \scite{Boyd2004}; $V \in \mathscr{P}_{2N}$ iff these equations hold.
By Lemma~\ref{lem:V(M)}, $V \in \varphi^{-1}(M)$ iff $M^{1/2} V M^{1/2} = \tilde{X} \oplus \tilde{Y}$ for some $\tilde{X},\tilde{Y} \in \mathscr{P}_N$.
In the notation of Eq.~(\ref{eq:M_sqrt_inv}), one has
\beq
M^{1/2} = \begin{bmatrix}
\tilde{C}_1 & -\tilde{S} \, \\ -\tilde{S}^T & \tilde{C}_2
\end{bmatrix} .
\eeq
Conjugating the equation $M^{1/2} V M^{1/2} = \tilde{X} \oplus \tilde{Y}$ by the positive definite matrix $\tilde{C}_1^{-1} \oplus \tilde{C}_2^{-1}$, it becomes
\beq
\label{eq:app_VF}
\begin{bmatrix}
\mathbb{I}_N & \! -F \\ -F^T & \! \mathbb{I}_N
\end{bmatrix} \!
\begin{bmatrix}
V_{RR} & V_{RL} \\ V_{RL}^T & V_{LL}
\end{bmatrix} \!
\begin{bmatrix}
\mathbb{I}_N & \! -F \\ -F^T & \! \mathbb{I}_N
\end{bmatrix}
=
\begin{bmatrix}
X & 0 \\ 0 & Y
\end{bmatrix} ,
\eeq
where $X \equiv \tilde{C}_1^{-1} \tilde{X} \tilde{C}_1^{-1}$ and $Y \equiv \tilde{C}_2^{-1} \tilde{Y} \tilde{C}_2^{-1}$.
These maps from $\tilde{X},\tilde{Y} \in \mathscr{P}_N$ to $X,Y \in \mathscr{P}_N$ are bijections.
Therefore, $V \in \varphi^{-1}(M)$ iff Eq.~(\ref{eq:app_VF}) holds for some $X,Y \in \mathscr{P}_N$.
The off-diagonal block of Eq.~(\ref{eq:app_VF}) yields Eq.~(\ref{eq:V_M_cond}).
The diagonal blocks of Eq.~(\ref{eq:app_VF}) are automatically satisfied, because the matrix on the left side is positive definite (it was constructed by conjugating $V \in \mathscr{P}_{2N}$ by other matrices in $\mathscr{P}_{2N}$).
\end{proof}

\subsection{Restrictions on the scaling dimension matrix $M$ imposed by 
%time-reversal ($\mathcal{T}$) and/or inversion ($\mathcal{P}$) 
symmetries}

Let $P$ be any permutation matrix that satisfies $P^2 = \mathbb{I}_{2N}$ and $PKP = -K$.
Let $\mathscr{S}_P \equiv \{ A \in \mathbb{R}^{2N \times 2N} \mid PAP = A \}$.
Then, the following results hold:

\begin{lem}
If $V \in \mathscr{P}_{2N} \cap \mathscr{S}_P$ and $M = \varphi(V)$, then $M \in \mathscr{S}_P$.
\end{lem}
\begin{proof}
Pick some $A \in SO(N,N)$ such that $A V \! A^T = D$ is diagonal, and define $B \equiv P A P$. 
It is easy to check that $B K B^T = K$ and $\det B = 1$, so $B \in SO(N,N)$. 
Also, $B V \! B^T = P D P$ is diagonal, since $P$ is a permutation matrix.
Thus, $M = B^T \! B = P A^T \! A \, P = P M P$.
\end{proof}

\begin{lem}\label{lem:V(M)_sym}
Assume $M \in \mathscr{M}_N \cap \mathscr{S}_P$.
Then $V \in \varphi^{-1}(M) \cap \mathscr{S}_P$ if and only if $V = M^{-1/2} Z M^{-1/2}$ for some $Z \in (\mathscr{P}_N \times \mathscr{P}_N) \cap \mathscr{S}_P$.
\end{lem}
\begin{proof}
Lemma~\ref{lem:V(M)} implies that $V \in \varphi^{-1}(M)$ iff $V$ has the specified form with $Z \in \mathscr{P}_N \times \mathscr{P}_N$.
Note that $P M^{\pm 1/2} = M^{\pm 1/2} P$.
Thus $V$ and $Z$ are conjugates of one another by an invertible matrix that commutes with $P$.
It follows that $PVP = V$ iff $PZP = Z$.
\end{proof}

Taking $Z = \mathbb{I}_{2N}$ shows that $\varphi^{-1}(M) \cap \mathscr{S}_P$ is nonempty for any  $M \in \mathscr{M}_N \cap \mathscr{S}_P$.
Thus, the set of interaction matrices $V$ that satisfy the constraint $P V P = V$ maps (under $\varphi$) \emph{onto} the set of scaling dimension matrices $M$ that satisfy the constraint $P M P = M$.
The constraints on $V$ derived in Section~\ref{app:V_sym_restrictions} are precisely of the form $PVP = V$ (with $P = \Sigma$, $\Sigma_1$ or $\Sigma_2$).
Hence the allowed scaling dimension matrices $M$ for a system with time-reversal ($\mathcal{T}$) and/or spatial inversion ($\mathcal{P}$) symmetry may be characterized as follows.

\subsubsection{$\mathcal{T}$ symmetry or $\mathcal{P}$ symmetry, but not both}
\label{app:M_sym_TorP}

First consider the case in which the system has either time-reversal symmetry or inversion symmetry, but not both.
We choose the $2N$ chiral bosons $\phi_I$ to transform according to Eq.~(\ref{eq:TR_def}) in the former case, and according to Eq.~(\ref{eq:inv_def}) in the latter. 
Then, in either case, the scaling dimension matrix must satisfy
\beq
\label{eq:M_sym1_cond}
\Sigma M \Sigma = M ,
\eeq
where $\Sigma \equiv \sigma_x \otimes \mathbb{I}_N$ (with no further constraints).
Imposing this constraint on the hyperbolic CS decomposition of $M$, Eq.~(\ref{eq:app_hyp_CS}), yields the conditions
\begin{subequations}
\begin{align}
Q C Q^T = C , \\*
Q S Q = S ,
\end{align}
\end{subequations}
where $Q \equiv Q_1 Q_2^T$.
Assume, without loss of generality, that there are $m$ distinct rapidities $\lambda_{\alpha}$ with multiplicities $N_{\alpha}$ (satisfying $N_1 + \cdots + N_m = N$), ordered so that $C = \cosh L$, $S = \sinh L$, and 
$L = \lambda_1 \mathbb{I}_{N_1} \oplus \lambda_2 \mathbb{I}_{N_2} \oplus \cdots \oplus \lambda_m \mathbb{I}_{N_m}$.
Then the conditions above require that
\beq
Q = R_1 \oplus R_2 \oplus \cdots \oplus R_m ,
\eeq
where $R_{\alpha} \in O(N_{\alpha})$ and $R_{\alpha} = R_{\alpha}^T$ (i.e.~each $R_{\alpha}$ is an $N_{\alpha} \times N_{\alpha}$ reflection matrix).
In the special case in which all rapidities are equal, one has
\beq
M = \begin{bmatrix}
\mathbb{I}_N \cosh{\lambda} & -R \sinh{\lambda} \\
-R \sinh{\lambda} & \mathbb{I}_N \cosh{\lambda}
\end{bmatrix} ,
\eeq
where $R \in O(N)$ and $R = R^T$.

\subsubsection{Both $\mathcal{T}$ and $\mathcal{P}$ symmetry}
\label{app:M_sym_TorP}

In the case that the system has both time-reversal and inversion symmetry, we choose the $4N$ chiral bosons $\phi_I$ to transform according to Eq.~(\ref{eq:TR_inv_def}). 
Then, the scaling dimension matrix must satisfy
\begin{subequations}
\label{eq:M_sym2_cond}
\begin{align}
(\mathcal{T}) \qquad \Sigma_1 M \Sigma_1 &= M , \\*
(\mathcal{P}) \qquad \Sigma_2 M \Sigma_2 &= M ,
\end{align}
\end{subequations}
where $\Sigma_1 \equiv \sigma_x \otimes \mathbb{I}_2 \otimes \mathbb{I}_N$ and
$\Sigma_2 \equiv \sigma_x \otimes \sigma_x \otimes \mathbb{I}_N$ (with no further independent constraints).
We can again impose these constraints on the hyperbolic CS decomposition of $M$, Eq.~(\ref{eq:app_hyp_CS}), to obtain an explicit parameterization of all scaling dimension matrices that are consistent with both $\mathcal{T}$ and $\mathcal{P}$ symmetry. However, the general parameterization is somewhat cumbersome, so we will omit it. 
In the special case in which all rapidities are equal, one finds
\beq
M = \begin{bmatrix}
\mathbb{I}_{2N} \cosh{\lambda} & -R \sinh{\lambda} \\
-R \sinh{\lambda} & \mathbb{I}_{2N} \cosh{\lambda}
\end{bmatrix} ,
\eeq
where $R \in O(2N)$ has the block form
\beq
R = \begin{bmatrix}
R_1 & R_2 \\ R_2 & R_1
\end{bmatrix} ,
\eeq
with $R_i = R_i^T$ (i.e.~$R$ is a $2N \times 2N$ reflection matrix with this particular block form).

\subsection{Intuition for the parameterization of $V$, Eq.~(\ref{eq:V_param})}

To gain some intuition for the parametrization (\ref{eq:V_param}) of $V$, first consider the limit $M = \mathbb{I}_{2N}$.
Then the interactions encoded in $X$ simply mix the right-movers amongst themselves, leading to new modes with renormalized velocities, while $Y$ does the same with the left-movers.
All scaling dimensions (being determined by $M$ alone) remain equal to their values at the free fixed point.
Next consider a different limit, $X = Y = \mathbb{I}_N$.
Now $V = M^{-1}$ is itself in $SO(N,N)$, and its inverse gives the scaling dimensions directly.
To connect these two limits, consider the Euclidean space of all symmetric $N \times N$ matrices, $\mathbb{R}^{N(N+1)/2}$. 
The positive definite matrices occupy the interior of a convex cone $\mathscr{P}_N \subset \mathbb{R}^{N(N+1)/2}$.
The space of interaction matrices is $\mathscr{V} = \mathscr{P}_{2N}$.
According to Eq.~(\ref{eq:V_param}), $\mathscr{V}$ should be regarded as a bundle of lower-dimensional convex cones $\mathscr{P}_N \times \mathscr{P}_N$ (parameterized by $X,Y$) as fibers over the $N^2$-dimensional submanifold $\mathscr{M}_N \equiv SO(N,N) \cap \mathscr{P}_{2N}$ (parameterized by $M$).
The scaling dimensions $\Delta(\mathbf{m})$, regarded as functions from $\mathscr{V} \to \mathbb{R}$, are then constant on each fiber.
Each Luttinger liquid phase, defined in terms of its instabilities (or lack thereof), thus extends over the interior of a solid cone emanating from the vertex of $\mathscr{V}$.

\subsection{Illustration of parameterization for $N=1$ channel}

The $N=1$ channel case again provides a nice illustration of these general ideas.
The set $\mathscr{V} = \mathscr{P}_2$ consists of all $2 \times 2$ matrices
\beq
\label{eq:V1}
V = \begin{bmatrix}
\alpha + \beta & \gamma \\ \gamma & \alpha - \beta
\end{bmatrix}
\eeq
with $(\alpha,\beta,\gamma) \in \mathbb{R}^3$ and $\alpha > (\beta^2 + \gamma^2)^{1/2}$.
This is quite clearly the interior of a circular cone in $\mathbb{R}^3$.
The parameterization (\ref{eq:V_param}), with $M = e^{-\lambda \sigma_x} \in \mathscr{M}_1$, corresponds to
\begin{subequations}
\begin{align}
\alpha &= \tfrac{1}{2}(x + y) \cosh{\lambda} , \\*
\beta &= \tfrac{1}{2}(x - y) , \\*
\gamma &= \tfrac{1}{2} (x + y) \sinh{\lambda} ,
\end{align}
\end{subequations}
where $x,y > 0$.
For fixed $\lambda$, the image of the resulting map $(x,y) \mapsto (\alpha,\beta,\gamma)$ is a slice of the cone $\mathscr{P}_2$, in a plane parallel to the $\beta$-axis and at an angle $\arctan(\tanh{\lambda})$ from the $\alpha$-axis.
Each such slice is the interior of a cone in $\mathbb{R}^2$, with an opening angle that decreases with increasing $\abs{\lambda}$. 
In terms of stability with respect to clean SC and CDW perturbations, $\mathscr{V} = \mathscr{P}_2$ splits into four regions: $\lambda < -\log 2$ ($\Delta_{\text{SC}} < 2 < \Delta_{\text{CDW}}$), $-\log 2 < \lambda < 0$ ($\Delta_{\text{SC}} < \Delta_{\text{CDW}} < 2$), $0 < \lambda < \log 2$ ($\Delta_{\text{CDW}} < \Delta_{\text{SC}} < 2$), and $\lambda > \log 2$ ($\Delta_{\text{CDW}} < 2 < \Delta_{\text{SC}}$).
These regions indeed take the form of solid cones emanating from the vertex of $\mathscr{P}_2$, as illustrated in Figure~\ref{fig:1channel_pd}.

\begin{figure}[t]
    \centering
    \hspace{-1em}
    \includegraphics[width=0.682\columnwidth]{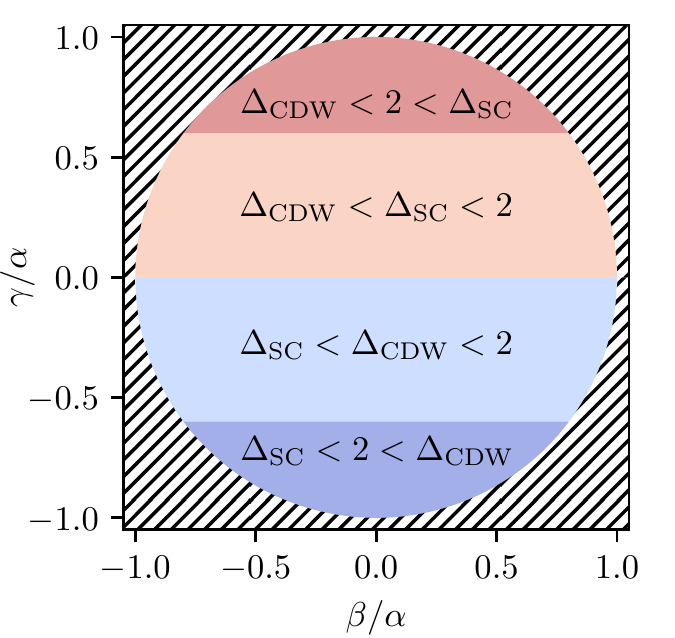}
    \caption{Phase diagram for the $N=1$ channel Luttinger liquid, in terms of stability with respect to global SC and CDW perturbations.
    The hatched region is unphysical ($V$ is not positive definite for these parameter values).}
    \label{fig:1channel_pd}
\end{figure}

\clearpage
\onecolumngrid

\section{Explicit parameterization of matrices for $N=2$ channel Luttinger liquid}

\subsection{Scaling dimension matrix $M$}
\label{sec:M2_param}

Let $Q(\phi)$ denote the $SO(2)$ rotation matrix
\beq
\label{eq:Q_def}
Q(\phi) \equiv \begin{bmatrix}
\cos{\phi} & \sin{\phi} \\ -\sin{\phi} & \cos{\phi}
\end{bmatrix} ,
\eeq
let $R(\phi)$ denote the $O(2)$ reflection matrix
\beq
\label{eq:R_def}
R(\phi) \equiv \begin{bmatrix}
\, \cos{\phi} & \sin{\phi} \\ \, \sin{\phi} & -\cos{\phi}
\end{bmatrix} ,
\eeq
and let
\beq
L \equiv \begin{bmatrix}
\delta + \lambda & 0 \\ 0 & \delta - \lambda
\end{bmatrix} .
\eeq
The parameterization of $M \in \mathscr{M}_2 \equiv SO(2,2) \cap \mathscr{P}_4$ described in Eqs.~(7) and (9) of the main text corresponds to
\beq
\label{eq:M2_param_1}
M = \begin{bmatrix}
Q^T\!(\tfrac{\theta - \alpha}{2}) & 0 \\ 0 & Q^T\!(\tfrac{\theta + \alpha}{2}) 
\end{bmatrix}
\begin{bmatrix}
\cosh{L} & -\sinh{L} \\ -\sinh{L} & \cosh{L}
\end{bmatrix}
\begin{bmatrix}
Q(\tfrac{\theta - \alpha}{2}) & 0 \\ 0 & Q(\tfrac{\theta + \alpha}{2})
\end{bmatrix} ,
\eeq
where each entry is a $2 \times 2$ matrix. Performing the matrix multiplication, we can write the result as
\beq
\label{eq:M2_param_2}
M = \begin{bmatrix}
\mathbb{I}_2 \cosh{\lambda} & -R(\theta) \sinh{\lambda} \\
-R(\theta) \sinh{\lambda} & \mathbb{I}_2 \cosh{\lambda}
\end{bmatrix} \cosh{\delta}
+ \begin{bmatrix}
R(\theta - \alpha) \sinh{\lambda} & -Q(\alpha) \cosh{\lambda} \\
-Q^T\!(\alpha) \cosh{\lambda} & R(\theta + \alpha) \sinh{\lambda}
\end{bmatrix} \sinh{\delta} ,
\eeq
where $\mathbb{I}_2$ is the $2 \times 2$ identity matrix.
The dependence on $\alpha$ disappears in the limit $\delta \to 0$, as stated in the main text.
Similarly, the dependence on $\theta$ disappears in the limit $\lambda \to 0$.

\subsubsection{$\mathcal{T}$ symmetry or $\mathcal{P}$ symmetry, but not both}

If the system has either time-reversal symmetry or inversion symmetry, but not both, the scaling dimension matrix must satisfy $\Sigma M \Sigma = M$, where $\Sigma \equiv \sigma_x \otimes \mathbb{I}_2$.
Applied to Eq.~(\ref{eq:M2_param_2}), this condition requires $\alpha = 0$.
Hence, in this case,
\begin{subequations}
\label{eq:M2_param_sym1}
\begin{align}
M &= \begin{bmatrix}
Q^T\!(\theta/2) & 0 \\ 0 & Q^T\!(\theta/2) 
\end{bmatrix}
\begin{bmatrix}
\cosh{L} & -\sinh{L} \\ -\sinh{L} & \cosh{L}
\end{bmatrix}
\begin{bmatrix}
Q(\theta/2) & 0 \\ 0 & Q(\theta/2)
\end{bmatrix} \\*[0.2em]
&= \begin{bmatrix}
\mathbb{I}_2 \cosh{\lambda} & -R(\theta) \sinh{\lambda} \\
-R(\theta) \sinh{\lambda} & \mathbb{I}_2 \cosh{\lambda}
\end{bmatrix} \cosh{\delta}
+ \begin{bmatrix}
R(\theta) \sinh{\lambda} & -\mathbb{I}_2 \cosh{\lambda} \\
-\mathbb{I}_2 \cosh{\lambda} & R(\theta) \sinh{\lambda}
\end{bmatrix} \sinh{\delta} .
\end{align}
\end{subequations}
Note in particular that the presence of $\mathcal{T}$ symmetry or $\mathcal{P}$ symmetry (but not both simultaneously) places \emph{no restrictions} on the allowed values of the parameters $\lambda$, $\delta$, and $\theta$.
Changing the sign of $\lambda$ is equivalent to shifting $\theta$ by $\pi$, so we can assume $\lambda > 0$.

\subsubsection{Both $\mathcal{T}$ and $\mathcal{P}$ symmetry}

If the system has both time-reversal and inversion symmetry, the scaling dimension matrix must satisfy $\Sigma_i M \Sigma_i = M$ for $i = 1,2$, where $\Sigma_1 \equiv \sigma_x \otimes \mathbb{I}_2$ and $\Sigma_2 \equiv \sigma_x \otimes \sigma_x$.
Applied to Eq.~(\ref{eq:M2_param_2}), these conditions require $\alpha = 0$ and $\theta = \pi/2$.
Hence, in this case,
\begin{subequations}
\label{eq:M2_param_sym2}
\begin{align}
M &= \begin{bmatrix}
Q^T\!(\pi/4) & 0 \\ 0 & Q^T\!(\pi/4) 
\end{bmatrix}
\begin{bmatrix}
\cosh{L} & -\sinh{L} \\ -\sinh{L} & \cosh{L}
\end{bmatrix}
\begin{bmatrix}
Q(\pi/4) & 0 \\ 0 & Q(\pi/4)
\end{bmatrix} \\*[0.2em]
&= \begin{bmatrix}
\mathbb{I}_2 \cosh{\lambda} & -\sigma_x \sinh{\lambda} \\
-\sigma_x \sinh{\lambda} & \mathbb{I}_2 \cosh{\lambda}
\end{bmatrix} \cosh{\delta}
+ \begin{bmatrix}
\sigma_x \sinh{\lambda} & -\mathbb{I}_2 \cosh{\lambda} \\
-\mathbb{I}_2 \cosh{\lambda} & \sigma_x \sinh{\lambda}
\end{bmatrix} \sinh{\delta} .
\end{align}
\end{subequations}

\clearpage
\subsection{Interaction matrix $V$ (general expressions)}

We parameterize the interaction matrix $V \in \mathscr{P}_4$ using Lemma~\ref{lem:V(M)_alt}.
Let
\beq
V = \begin{bmatrix}
V_{RR} & V_{RL} \\ V_{RL}^T & V_{LL}
\end{bmatrix} .
\eeq
Lemma~\ref{lem:V(M)_alt} states that $V \in \varphi^{-1}(M)$ if and only if
\begin{subequations}
\label{eq:V2_param_eqs}
\begin{align}
V_{RR} &= X + F \, Y F^T , \\*
V_{LL} &= Y + F^T \! X F , \\*
V_{RL} &= X F + F \, Y
\end{align}
\end{subequations}
for some $X,Y \in \mathscr{P}_2$.
The $2 \times 2$ matrix $F$ corresponding to the $M$ given in Eqs.~(\ref{eq:M2_param_1}) or (\ref{eq:M2_param_2}) is
\beq
\label{eq:F2}
F = Q^T\!(\tfrac{\theta - \alpha}{2}) 
\begin{bmatrix}
\tanh(\tfrac{\delta+\lambda}{2}) & 0 \\ 0 & \tanh(\tfrac{\delta-\lambda}{2}) 
\end{bmatrix} 
Q(\tfrac{\theta + \alpha}{2}) 
= \frac{\sinh{\lambda}}{\cosh{\lambda} + \cosh{\delta}} \, R(\theta) + \frac{\sinh{\delta}}{\cosh{\lambda} + \cosh{\delta}} \, Q(\alpha) ,
\eeq
where $Q(\phi)$ and $R(\phi)$ are defined in Eqs.~(\ref{eq:Q_def}) and (\ref{eq:R_def}).

The matrices $X,Y \in \mathscr{P}_2$ can be conveniently parameterized as
\begin{subequations}
\label{eq:XY2}
\begin{align}
X &= \begin{bmatrix}
x_0 + x_1 & x_2 \\ x_2 & x_0 - x_1
\end{bmatrix} , \\*
Y &= \begin{bmatrix}
y_0 + y_1 & \ y_2 \\ y_2 & \ y_0 - y_1
\end{bmatrix} ,
\end{align}
\end{subequations}
where $(x_0,x_1,x_2) \in \mathbb{R}^3$,  $x_0 > (x_1^2 + x_2^2)^{1/2}$, and $(y_0,y_1,y_2) \in \mathbb{R}^3$, $y_0 > (y_1^2 + y_2^2)^{1/2}$. 
Using Eqs.~(\ref{eq:F2}) and (\ref{eq:XY2}) in Eq.~(\ref{eq:V2_param_eqs}), we obtain
\small
\begin{subequations}
\label{eq:V2_param_full}
\begin{align}
V_{RR} &= 
\begin{bmatrix}
x_0 + x_1 & x_2 \\ x_2 & x_0 - x_1
\end{bmatrix}
+ \frac{\sinh^2\!\lambda}{(\cosh{\lambda} + \cosh{\delta})^2}
\begin{bmatrix}
y_0 + y_1 \cos{2\theta} + y_2 \sin{2\theta} & y_1 \sin{2\theta} - y_2 \cos{2\theta} \\ 
y_1 \sin{2\theta} - y_2 \cos{2\theta} & y_0 - y_1 \cos{2\theta} - y_2 \sin{2\theta}
\end{bmatrix} \nonumber \\*[0.5em]
&\quad + 
\frac{2\sinh{\lambda} \sinh{\delta}}{(\cosh{\lambda} + \cosh{\delta})^2}
\begin{bmatrix}
y_1 \cos(\theta+\alpha) + y_2 \sin(\theta+\alpha) + y_0 \cos(\theta-\alpha) & y_0 \sin(\theta-\alpha) \\ 
y_0 \sin(\theta-\alpha) & y_1 \cos(\theta+\alpha) + y_2 \sin(\theta+\alpha) - y_0 \cos(\theta-\alpha)
\end{bmatrix} \nonumber \\*[0.5em]
&\quad + 
\frac{\sinh^2\!\delta}{(\cosh{\lambda} + \cosh{\delta})^2}
\begin{bmatrix}
y_0 + y_1 \cos{2\alpha} + y_2 \sin{2\alpha} & y_2 \cos{2\alpha} - y_1 \sin{2\alpha} \\ 
y_2 \cos{2\alpha} - y_1 \sin{2\alpha} & y_0 - y_1 \cos{2\alpha} - y_2 \sin{2\alpha}
\end{bmatrix} , 
\label{eq:VRR_param_full}
\\[1em]
V_{LL} &= 
\begin{bmatrix}
y_0 + y_1 & y_2 \\ y_2 & y_0 - y_1
\end{bmatrix}
+ \frac{\sinh^2\!\lambda}{(\cosh{\lambda} + \cosh{\delta})^2}
\begin{bmatrix}
x_0 + x_1 \cos{2\theta} + x_2 \sin{2\theta} & x_1 \sin{2\theta} - x_2 \cos{2\theta} \\ 
x_1 \sin{2\theta} - x_2 \cos{2\theta} & x_0 - x_1 \cos{2\theta} - x_2 \sin{2\theta}
\end{bmatrix} \nonumber \\*[0.5em]
&\quad + 
\frac{2\sinh{\lambda} \sinh{\delta}}{(\cosh{\lambda} + \cosh{\delta})^2}
\begin{bmatrix}
x_1 \cos(\theta-\alpha) + x_2 \sin(\theta-\alpha) + x_0 \cos(\theta+\alpha) & x_0 \sin(\theta+\alpha) \\ 
x_0 \sin(\theta+\alpha) & x_1 \cos(\theta-\alpha) + x_2 \sin(\theta-\alpha) - x_0 \cos(\theta+\alpha)
\end{bmatrix} \nonumber \\*[0.5em]
&\quad + 
\frac{\sinh^2\!\delta}{(\cosh{\lambda} + \cosh{\delta})^2}
\begin{bmatrix}
x_0 + x_1 \cos{2\alpha} - x_2 \sin{2\alpha} & x_1 \sin{2\alpha} + x_2 \cos{2\alpha} \\ 
x_1 \sin{2\alpha} + x_2 \cos{2\alpha} & x_0 - x_1 \cos{2\alpha} + x_2 \sin{2\alpha}
\end{bmatrix} , 
\label{eq:VLL_param_full}
\\[1em]
V_{RL} &= 
\frac{\sinh{\lambda}}{\cosh{\lambda} + \cosh{\delta}}
\begin{bmatrix}
(x_1 + y_1 + x_0 + y_0) \cos{\theta} + (x_2 + y_2) \sin{\theta} & 
(x_0 + y_0 + x_1 - y_1) \sin{\theta} - (x_2 - y_2) \cos{\theta} \\ 
(x_0 + y_0 - x_1 + y_1) \sin{\theta} + (x_2 - y_2) \cos{\theta} & 
(x_1 + y_1 - x_0 - y_0) \cos{\theta} + (x_2 + y_2) \sin{\theta}
\end{bmatrix} \nonumber \\*[0.5em]
&\quad + 
\frac{\sinh{\delta}}{\cosh{\lambda} + \cosh{\delta}}
\begin{bmatrix}
(x_0 + y_0 + x_1 + y_1) \cos{\alpha} - (x_2 - y_2) \sin{\alpha} & 
(x_2 + y_2) \cos{\alpha} + (x_1 - y_1 + x_0 + y_0) \sin{\alpha} \\ 
(x_2 + y_2) \cos{\alpha} + (x_1 - y_1 - x_0 - y_0) \sin{\alpha} & 
(x_0 + y_0 - x_1 - y_1) \cos{\alpha} + (x_2 - y_2) \sin{\alpha}
\end{bmatrix} .
\label{eq:VRL_param_full}
\end{align}
\end{subequations}
\normalsize

Equation~(\ref{eq:V2_param_full}) gives a complete and explicit parameterization of the possible interaction matrices $V \in \mathscr{P}_4$ of a $2$-channel Luttinger liquid, in terms of the ten real parameters $(\lambda,\delta,\theta,\alpha,x_0,x_1,x_2,y_0,y_1,y_2)$. 
Of these, only the first four $(\lambda,\delta,\theta,\alpha)$ affect scaling dimensions; they determine the scaling dimension matrix $M$ via Eq.~(\ref{eq:M2_param_2}). 
The remaining six parameters can be chosen arbitrarily, subject only to the constraints $x_0 > (x_1^2 + x_2^2)^{1/2}$ and $y_0 > (y_1^2 + y_2^2)^{1/2}$ (if either of these inequalities is violated, the resulting $V$ will fail to be positive definite).

\subsubsection{$\mathcal{T}$ symmetry or $\mathcal{P}$ symmetry, but not both}

When symmetries are present, it is convenient to parameterize the interaction matrix $V \in \mathscr{P}_4$ using Lemma~\ref{lem:V(M)_sym} instead.
In the present case, Lemma~\ref{lem:V(M)_sym} gives $V = M^{-1/2} [X \oplus Y] M^{-1/2}$ with $X,Y \in \mathscr{P}_2$ satisfying $\Sigma [X \oplus Y] \Sigma = X \oplus Y$, where $\Sigma \equiv \sigma_x \otimes \mathbb{I}_2$;
the condition fixes $Y = X$.
Thus, $V = M^{-1/2} [X \oplus X] M^{-1/2}$ with $X \in \mathscr{P}_2$.
The scaling dimension matrix $M$ is given by Eq.~(\ref{eq:M2_param_sym1}).
Thus,
\beq
M^{-1/2} = \begin{bmatrix}
Q^T\!(\theta/2) & 0 \\ 0 & Q^T\!(\theta/2) 
\end{bmatrix}
\begin{bmatrix}
\cosh(L/2) & \sinh(L/2) \\ \sinh(L/2) & \cosh(L/2)
\end{bmatrix}
\begin{bmatrix}
Q(\theta/2) & 0 \\ 0 & Q(\theta/2)
\end{bmatrix} .
\eeq
The inner factors of $Q(\theta/2)$ in the product $V = M^{-1/2} [X \oplus X] M^{-1/2}$ may be absorbed into $X$, since the latter is an arbitrary element of $\mathscr{P}_2$.
Doing so, we obtain
\beq
\label{eq:V2_sym1_1}
V = \begin{bmatrix}
Q^T\!(\theta/2) & 0 \\ 0 & Q^T\!(\theta/2) 
\end{bmatrix}
\begin{bmatrix}
\cosh(L/2) & \sinh(L/2) \\ \sinh(L/2) & \cosh(L/2)
\end{bmatrix}
\begin{bmatrix}
X & 0 \\ 0 & X
\end{bmatrix}
\begin{bmatrix}
\cosh(L/2) & \sinh(L/2) \\ \sinh(L/2) & \cosh(L/2)
\end{bmatrix}
\begin{bmatrix}
Q(\theta/2) & 0 \\ 0 & Q(\theta/2)
\end{bmatrix} .
\eeq
We can parameterize $X \in \mathscr{P}_2$ as
\beq
X = \zeta \begin{bmatrix}
1 + a & -b \\ -b & 1 - a
\end{bmatrix} , 
\eeq
where $\zeta > 0$ and $(a,b) \in \mathbb{R}^2$,  $a^2 + b^2 < 1$.
Performing the matrix multiplications in Eq.~(\ref{eq:V2_sym1_1}), we obtain
\beq
\label{eq:V2_param_sym1_1}
V = \zeta \begin{bmatrix}
V_1 & V_2 \\ V_2 & V_1
\end{bmatrix} \cosh{\delta}
+ \zeta \begin{bmatrix}
V_2 & V_1 \\ V_1 & V_2
\end{bmatrix} \sinh{\delta} ,
\eeq
where 
\begin{subequations}
\label{eq:V2_param_sym1_2}
\begin{align}
V_1 &= \begin{bmatrix}
\cosh{\lambda} + a \cos{\theta} \cosh{\lambda} + b \sin{\theta} & 
a \sin{\theta} \cosh{\lambda} - b \cos{\theta} \\ 
a \sin{\theta} \cosh{\lambda} - b \cos{\theta} & 
\cosh{\lambda} - a \cos{\theta} \cosh{\lambda} - b \sin{\theta}
\end{bmatrix} , \\*
V_2 &= \begin{bmatrix}
a + \cos{\theta} & 
\sin{\theta} \\
\sin{\theta} &
a - \cos{\theta}
\end{bmatrix} \sinh{\lambda} .
\end{align}
\end{subequations}

Equations~(\ref{eq:V2_param_sym1_1}) and (\ref{eq:V2_param_sym1_2}) gives a complete and explicit parameterization of the possible interaction matrices $V \in \mathscr{P}_4$ of a $2$-channel Luttinger liquid with either time-reversal symmetry or inversion symmetry (but not both), in terms of the six real parameters $(\lambda,\delta,\theta,a,b,\zeta)$. 
Of these, only the first three $(\lambda,\delta,\theta)$ affect scaling dimensions; they determine the scaling dimension matrix $M$ via Eq.~(\ref{eq:M2_param_sym1}). 
The remaining three parameters can be chosen arbitrarily, subject only to the constraints $a^2 + b^2 < 1$ and $\zeta > 0$ (if these inequalities are violated, the resulting $V$ will fail to be positive definite).
Note that $\zeta$ is an irrelevant overall scale factor.

\subsubsection{Both $\mathcal{T}$ and $\mathcal{P}$ symmetry}

As above, we parameterize the interaction matrix $V \in \mathscr{P}_4$ using Lemma~\ref{lem:V(M)_sym}.
In this case, Lemma~\ref{lem:V(M)_sym} gives $V = M^{-1/2} [X \oplus Y] M^{-1/2}$ with $X,Y \in \mathscr{P}_2$ satisfying $\Sigma_i [X \oplus Y] \Sigma_i = X \oplus Y$ for $i =1,2$, where $\Sigma_1 \equiv \sigma_x \otimes \mathbb{I}_2$ and $\Sigma_2 \equiv \sigma_x \otimes \sigma_x$.
The conditions fix $Y = X = x_0 \mathbb{I}_2 + x_1 \sigma_x$, with $(x_0,x_1) \in \mathbb{R}^2$, $x_0 > \abs{x_1}$.
Thus, $V = M^{-1/2} [X \oplus X] M^{-1/2}$ with $X$ of the form specified.
The scaling dimension matrix $M$ is given by Eq.~(\ref{eq:M2_param_sym2}).
Thus,
\beq
M^{-1/2} = \begin{bmatrix}
Q^T\!(\pi/4) & 0 \\ 0 & Q^T\!(\pi/4) 
\end{bmatrix}
\begin{bmatrix}
\cosh(L/2) & \sinh(L/2) \\ \sinh(L/2) & \cosh(L/2)
\end{bmatrix}
\begin{bmatrix}
Q(\pi/4) & 0 \\ 0 & Q(\pi/4)
\end{bmatrix} .
\eeq
Writing $X = \zeta( \mathbb{I}_2 + a \sigma_x)$ with $\zeta > 0$, $\abs{a} < 1$, and performing the matrix multiplications in $V = M^{-1/2} [X \oplus X] M^{-1/2}$, we obtain
\beq
\label{eq:V2_param_sym2_1}
V = \zeta \begin{bmatrix}
V_1 & V_2 \\ V_2 & V_1
\end{bmatrix} \cosh{\delta}
+ \zeta \begin{bmatrix}
V_2 & V_1 \\ V_1 & V_2
\end{bmatrix} \sinh{\delta} ,
\eeq
where 
\begin{subequations}
\label{eq:V2_param_sym2_2}
\begin{align}
V_1 &= \begin{bmatrix}
1 & a \\ a & 1
\end{bmatrix} \cosh{\lambda} , \\*
V_2 &= \begin{bmatrix}
a & 1 \\ 1 & a
\end{bmatrix} \sinh{\lambda} .
\end{align}
\end{subequations}

Equations~(\ref{eq:V2_param_sym2_1}) and (\ref{eq:V2_param_sym2_2}) gives a complete and explicit parameterization of the possible interaction matrices $V \in \mathscr{P}_4$ of a $2$-channel Luttinger liquid with both time-reversal and inversion symmetry, in terms of the four real parameters $(\lambda,\delta,a,\zeta)$. 
Of these, only the first two $(\lambda,\delta)$ affect scaling dimensions; they determine the scaling dimension matrix $M$ via Eq.~(\ref{eq:M2_param_sym2}). 
The remaining two parameters can be chosen arbitrarily, subject only to the constraints $\abs{a} < 1$ and $\zeta > 0$ (if these inequalities are violated, the resulting $V$ will fail to be positive definite).
Note that, as before, $\zeta$ is an irrelevant overall scale factor.

\subsection{Interaction matrix $V$ with $\mathcal{T}$ or $\mathcal{P}$ symmetry (but not both), in the special case $\delta = 0$}

In the limit $\delta \to 0$, Eqs.~(\ref{eq:V2_param_sym1_1}) and (\ref{eq:V2_param_sym1_2}) together yield
\beq
\label{eq:V2_simple}
V = \left[ \begin{array}{cc|cc}
v_+ & \! w & c_+ & c_0 \\
w & v_- & c_0 & c_- \\[0.15em]
\hline
c_+ & c_0 & \, v_+ & \! w \\[-0.15em]
c_0 & c_- & w & v_-
\end{array} \right] ,
\eeq
where
\begin{subequations}
\label{eq:V2_simple_params}
\begin{align}
v_{\pm} &= v \pm u , \\*
v &= \zeta \cosh{\lambda} , \\*
u &= \zeta (a \cos{\theta} \cosh{\lambda} + b \sin{\theta}) , \\*
w &= \zeta (a \sin{\theta} \cosh{\lambda} - b \cos{\theta}) , \\* 
c_{\pm} &= \zeta (a \pm \cos{\theta}) \sinh{\lambda} , \\*
c_0 &= \zeta \sin{\theta} \sinh{\lambda} .
\end{align}
\end{subequations}
We assume $\lambda > 0$ (without loss of generality), and identify the values of $a,b$ for which all matrix elements of $V$ are nonnegative.
The diagonal elements of a positive definite matrix are necessarily positive, so $v_{\pm} > 0$ is automatic. 
$c_0 \geq 0$ requires $\theta \in [0,\pi]$.
Nonnegativity of the remaining matrix elements, $c_{\pm}$ and $w$, requires
\begin{subequations}
\begin{align}
a &\geq \abs{\cos{\theta}} , \label{eq:ab_ineq_1} \\*
a \sin{\theta} \cosh{\lambda} &\geq b \cos{\theta} . \label{eq:ab_ineq_2}
\end{align}
\end{subequations}
In fact, Eq.~(\ref{eq:ab_ineq_2}) is superfluous, because it follows from Eq.~(\ref{eq:ab_ineq_1}) and the positive definiteness condition $a^2 + b^2 < 1$;
assuming the latter, we have 
$\abs{b} < (1 - a^2)^{1/2} \leq \abs{\sin{\theta}} \leq \abs{\cos{\theta} \tan{\theta}} \cosh{\lambda} \leq a \abs{\tan{\theta}} \cosh{\lambda}$, which implies Eq.~(\ref{eq:ab_ineq_2}) if $\theta \in [0,\pi]$.
We conclude that the $V$ matrix given above is positive definite with all entries nonnegative if $\lambda > 0$, $\theta \in [0,\pi]$, $a^2 + b^2 < 1$, and $a \geq \abs{\cos{\theta}}$.

These results are completely equivalent to the ones stated in the main text.
Indeed, using Eq.~(\ref{eq:V2_simple_params}), one can verify that the following linear relations hold:
\begin{subequations}
\begin{align}
c_{\pm} &= (w \sin{\theta} \pm v_{\pm} \cos{\theta}) \tanh{\lambda} , \\*
c_0 &= v \sin{\theta} \tanh{\lambda} .
\end{align}
\end{subequations}
These are precisely the relations that one obtains by using Eq.~(\ref{eq:V2_simple}) and $F = Q(\theta) \tanh(\lambda/2)$ in Eq.~(\ref{eq:V_M_cond}) of Lemma~\ref{lem:V(M)_imp}.
Therefore, we can regard $(v,u,w)$, instead of $(\zeta,a,b)$, as the independent variables parameterizing $V$.
From Eq.~(\ref{eq:V2_simple_params}), we have
\beq
\begin{bmatrix}
u/v \\ w/v
\end{bmatrix}
= \begin{bmatrix}
\cos{\theta} & \sin{\theta} \sech{\lambda} \\ 
\sin{\theta} & -\cos{\theta} \sech{\lambda}
\end{bmatrix}
\begin{bmatrix}
a \\ b
\end{bmatrix} .
\eeq
Inverting this linear system,
\beq
\begin{bmatrix}
a \\ b
\end{bmatrix}
= \begin{bmatrix}
\cos{\theta} & \sin{\theta} \\ 
\, \sin{\theta} \cosh{\lambda} & -\cos{\theta} \cosh{\lambda}
\end{bmatrix} 
\begin{bmatrix}
u/v \\ w/v
\end{bmatrix} .
\eeq
Thus, the conditions $a^2 + b^2 < 1$ and $a \geq \abs{\cos{\theta}}$ translate to:
\begin{subequations}
\begin{align}
(u \sin{\theta} - w \cos{\theta})^2 \cosh^2\!\lambda + (u \cos{\theta} + w \sin{\theta})^2 &< v^2 , \\*
u \cos{\theta} + w \sin{\theta} &\geq v \abs{\cos{\theta}} .
\end{align}
\end{subequations}
These conditions are stated in the main text as Eq.~(12) and the in-line equation just below it.

\clearpage

\section{Absolute $q$-stability phase diagram for $N=2$ channel Luttinger liquid}

The interaction matrix $V \in \mathscr{P}_4$ of a $2$-channel Luttinger liquid depends on ten real parameters; these can be chosen according to Eq.~(\ref{eq:V2_param_full}).
Of the ten, only the four parameters $\xi \equiv (\lambda,\delta,\theta,\alpha)$ affect scaling dimensions; they determine the scaling dimension matrix $M(\xi)$ via Eq.~(\ref{eq:M2_param_2}), which in turn determines $\Delta(\mathbf{m};\xi) \equiv \tfrac{1}{2} \mathbf{m}^T \! M(\xi) \mathbf{m}$.

Let $q(\xi)$ denote the absolute $q$-stability value of a Luttinger liquid with parameters $\xi$.
By definition, $q(\xi)$ is the largest integer such that $\Delta(\mathbf{m};\xi) > 2$ for all nonzero $\mathbf{m} \in \mathbb{Z}^4$ with $K(\mathbf{m}) \in \mathbb{Z}$ and $\abs{\mathbf{m}} \leq q(\xi)$, where $K(\mathbf{m}) \equiv \tfrac{1}{2} \mathbf{m}^T \! K \mathbf{m}$, $K \equiv \mathbb{I}_2 \oplus - \mathbb{I}_2$.
Because of the inequality $\Delta(\mathbf{m}) \geq \abs{K(\mathbf{m})}$, one can restrict attention to those $\mathbf{m}$ for which $K(\mathbf{m}) = 0,\pm1,\pm2$.
Thus, an equivalent definition of $q(\xi)$ to the one given above is: $q(\xi)$ is the smallest positive integer such that $\Delta(\mathbf{m};\xi) \leq 2$ for some $\mathbf{m} \in \mathbb{Z}^4$ with $K(\mathbf{m}) \in \{0, \pm 1, \pm 2\}$ and $\abs{\mathbf{m}} = q(\xi)+1$.

We use the second definition to determine the phase diagram numerically.
The algorithm is straightforward:
at each point $\xi$, compute $\Delta(\mathbf{m};\xi)$ for all integer vectors $\mathbf{m}$ with $K(\mathbf{m}) \in \{0, \pm 1, \pm 2\}$ in shells of increasing $\abs{\mathbf{m}}$, until either a vector $\mathbf{m}_*$ is found for which $\Delta(\mathbf{m}_*;\xi) \leq 2$, or $\abs{\mathbf{m}}$ passes a specified cutoff value $q_*$.
Set $q(\xi) = \abs{\mathbf{m}_*}-1$ in the former case, and $q(\xi) = q_*$ in the latter.
The shells of vectors with $\abs{\mathbf{m}} =2,\dots, q_*$ can be tabulated in advance, and the matrix $M(\xi)$ only needs to be computed once at each point $\xi$.
Figure 2 of the main text and Figure~\ref{fig:absolute_q_stability} below were obtained by this method, with cutoff $q_* = 22$.
\\

\begin{figure*}[h]
\begin{subfigure}
    \centering
    \includegraphics{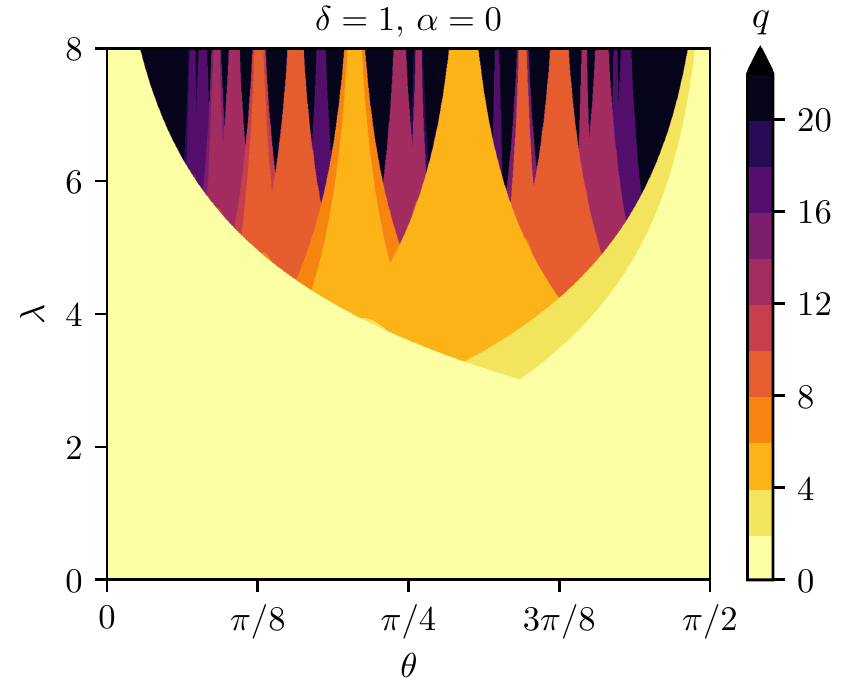}
\end{subfigure}
\
\begin{subfigure}
    \centering
    \includegraphics{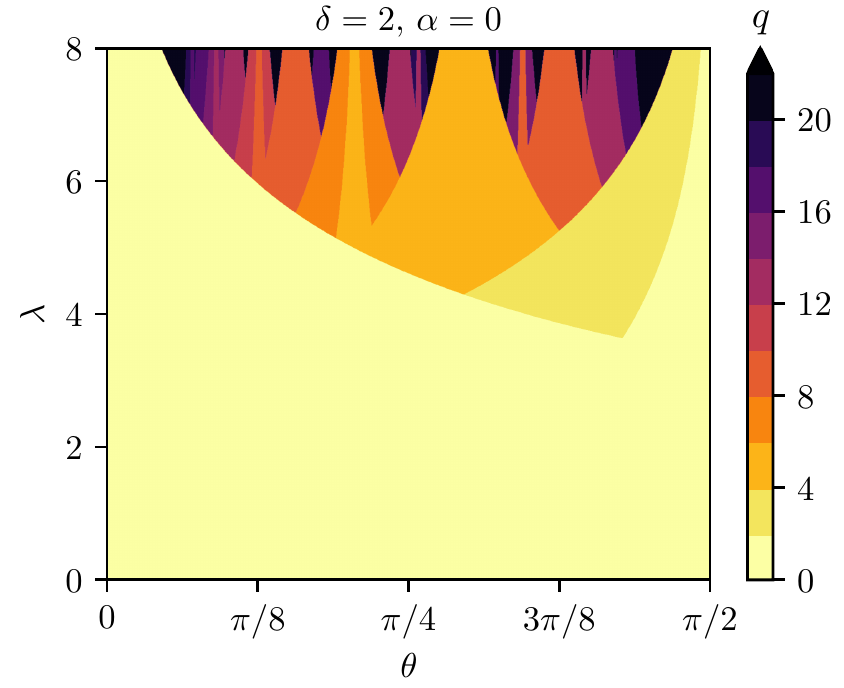}
\end{subfigure}\\
\vspace{0.5em}
\begin{subfigure}
    \centering
    \includegraphics{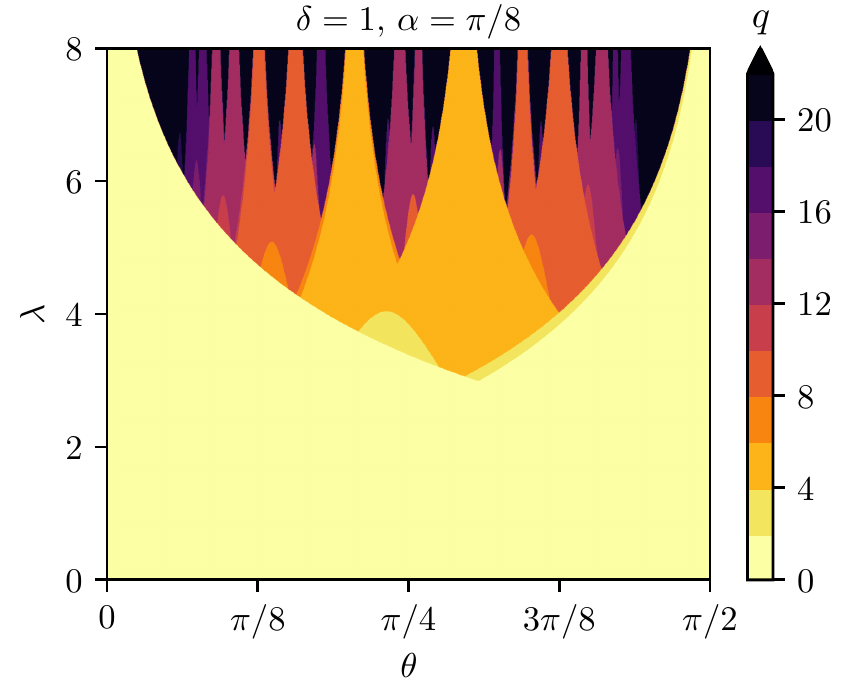}
\end{subfigure}
\
\begin{subfigure}
    \centering
    \includegraphics{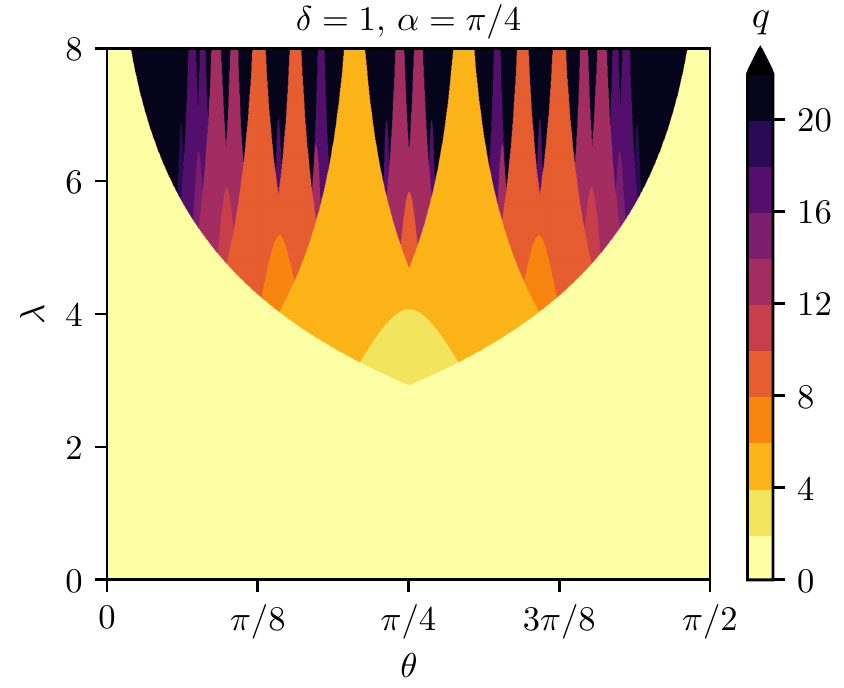}
\end{subfigure}
\caption{Various slices of the absolute $q$-stability phase diagram for the $N=2$ channel Luttinger liquid (compare Figure 2 of the main text).}
\label{fig:absolute_q_stability}
\end{figure*}

\clearpage
\twocolumngrid

\section{Representation of $2$-channel Luttinger liquid in terms of charge and spin fields}

A single-spinful-channel quantum wire provides the simplest example of a $2$-channel Luttinger liquid.
Standard treatments of this problem \cite{Giamarchi2003} are usually phrased in terms of non-chiral charge ($c$) and spin ($s$) fields, $\varphi_c$ and $\varphi_s$, and their canonical conjugates, $\Pi_c \equiv \frac{1}{\pi} \partial_x \vartheta_c$ and $\Pi_s \equiv \frac{1}{\pi} \partial_x \vartheta_s$ respectively.
These fields are related to the slowly varying parts of the charge density, $\rho_c$, and spin density, $\rho_s$, via $\rho_c = - \frac{1}{\pi} \partial_x \varphi_c$ and $\rho_s = - \frac{1}{\pi} \partial_x \varphi_s$ (we follow the normalization and sign conventions of Ref.~\cite{Giamarchi2003}).
Our analysis, meanwhile, is phrased in terms of chiral boson fields $\phi_I$, which are related to the densities at each Fermi point via $\rho_I = \frac{1}{2\pi} \partial_x \phi_I$.

A rigorous reformulation of the bosonic effective theory in terms of charge and spin fields is possible when the system is invariant under spin rotations about some axis $\mathbf{\hat{n}}$.
(Note that such a symmetry, by itself, imposes no constraints on the interaction matrix $V$.)
Then, the component of the spin along $\mathbf{\hat{n}}$ is a good quantum number, and it labels the different Fermi points.
One possibility for this labelling is
\beq
\label{eq:TR_spin}
(\phi_1, \phi_2, \phi_3, \phi_4) = (\phi_{R\uparrow}, \phi_{R\downarrow}, \phi_{L\downarrow}, \phi_{L\uparrow}) ,
\eeq
where $R/L$ distinguishes right movers from left movers, and $\uparrow / \downarrow$ denotes the spin component along $\mathbf{\hat{n}}$.
With the choice (\ref{eq:TR_spin}), the fields transform under time-reversal $\mathcal{T}$ according to Eq.~(\ref{eq:TR_def}).

Another possibility is to take
\beq
\label{eq:inv_spin}
(\phi_1, \phi_2, \phi_3, \phi_4) = (\phi_{R\uparrow}, \phi_{R\downarrow}, \phi_{L\uparrow}, \phi_{L\downarrow}) .
\eeq
With the choice (\ref{eq:inv_spin}), the fields transform under spatial inversion $\mathcal{P}$ according to Eq.~(\ref{eq:inv_def}).

In either case, one has
\beq
\begin{bmatrix}
\varphi_c \\ \vartheta_c \\ \varphi_s \\ \vartheta_s
\end{bmatrix}
= \frac{1}{2\sqrt{2}} \begin{bmatrix}
-1 & -1 & -1 & -1 \\
 1 &  1 & -1 & -1 \\
-1 &  1 & -1 &  1 \\
 1 & -1 & -1 &  1
\end{bmatrix}
\begin{bmatrix}
\phi_{R\uparrow} \\ \phi_{R\downarrow} \\ \phi_{L\uparrow} \\ \phi_{L\downarrow}
\end{bmatrix} ,
\eeq
and the inverse relation
\beq
\label{eq:charge_spin}
\begin{bmatrix}
\phi_{R\uparrow} \\ \phi_{R\downarrow} \\ \phi_{L\uparrow} \\ \phi_{L\downarrow}
\end{bmatrix}
= \frac{1}{\sqrt{2}} \begin{bmatrix}
-1 &  1 & -1 &  1 \\
-1 &  1 &  1 & -1 \\
-1 & -1 & -1 & -1 \\
-1 & -1 &  1 &  1
\end{bmatrix}
\begin{bmatrix}
\varphi_c \\ \vartheta_c \\ \varphi_s \\ \vartheta_s
\end{bmatrix} .
\eeq
Then, $\rho_s = -\frac{1}{\pi} \partial_x \varphi_s$ is indeed the slowly varying part of the density of excess spin in the $\mathbf{\hat{n}}$-direction.

If spin rotation symmetry is completely broken, on the other hand, one cannot easily reformulate the bosonic effective theory in terms of charge and spin fields.
One can of course still use the above formulae to define non-chiral fields $\varphi_s$ and $\vartheta_s$ as linear combinations of the $\phi_I$, but in general these non-chiral fields will have nothing to do with the physical spin.

Now consider a 2-channel Luttinger liquid with effective action specified by Eqs.~(2) and (10) of the main text:
\beq
S = \frac{1}{4\pi} \int dt \, dx \, \Big[ K_{IJ} \partial_t \phi_I \partial_x \phi_J - V_{IJ} \partial_x \phi_I \partial_x \phi_J \Big] ,
\eeq
where $K = \mathrm{diag}(-\mathbb{I}_2 , \mathbb{I}_2)$, and
\beq
\label{eq:V2_simple_cs}
V = \left[ \begin{array}{cc|cc}
v_+ & w & c_+ & c_0 \\
w & v_- & c_0 & c_- \\[0.15em]
\hline
c_+ & c_0 & v_+ & w \\[-0.15em]
c_0 & c_- & w & v_-
\end{array} \right] .
\eeq
To shorten subsequent expressions, let
\begin{subequations}
\begin{align}
v_{\pm} &\equiv v \pm u, \\*
c_{\pm} &\equiv c \pm b .
\end{align}
\end{subequations}
As discussed earlier, the $V$ matrix (\ref{eq:V2_simple_cs}) describes a system that has either time-reversal ($\mathcal{T}$) symmetry or spatial inversion ($\mathcal{P}$) symmetry, but not both.
Assuming that the system also has spin-rotation symmetry about some axis $\mathbf{\hat{n}}$ (as mentioned above, this assumption does not constrain $V$ at all), one can use Eqs.~(\ref{eq:TR_spin}--\ref{eq:charge_spin}) to write down the corresponding effective Hamiltonian $H$ in terms of charge and spin fields.

In the case of $\mathcal{T}$ symmetry, we use Eqs.~(\ref{eq:TR_spin}) and (\ref{eq:charge_spin}).
The result is
\begin{align}
H = \frac{1}{2\pi} \int dx \, \bigg[ 
&v_c K_c (\partial_x \vartheta_c)^2 + \frac{v_c}{K_c} (\partial_x \varphi_c)^2 \nonumber \\*
&+ v_s K_s (\partial_x \vartheta_s)^2 + \frac{v_s}{K_s} (\partial_x \varphi_s)^2 \nonumber \\*
&+ d_+ \partial_x \vartheta_c \partial_x \varphi_s + d_- \partial_x \vartheta_s \partial_x \varphi_c \bigg] ,
\end{align}
where
\begin{subequations}
\begin{align}
v_{c(s)} &\equiv [(v \pm w)^2 - (c \pm c_0)^2]^{1/2} , \\*
K_{c(s)} &\equiv \left[ \frac{v \pm w \mp c - c_0}{v \pm w \pm c + c_0} \right]^{1/2} , \\*
d_{\pm} &\equiv -2u \pm 2b .
\end{align}
\end{subequations}
Using $\Pi_c \equiv \frac{1}{\pi} \partial_x \vartheta_c$ and $\Pi_s \equiv \frac{1}{\pi} \partial_x \vartheta_s$, we recover Eq.~(13) in the main text.

In the case of $\mathcal{P}$ symmetry, we use Eqs.~(\ref{eq:inv_spin}) and (\ref{eq:charge_spin}) instead.
The result is then
\begin{align}
H = \frac{1}{2\pi} \int dx \, \bigg[ 
&v_c K_c (\partial_x \vartheta_c)^2 + \frac{v_c}{K_c} (\partial_x \varphi_c)^2 \nonumber \\*
&+ v_s K_s (\partial_x \vartheta_s)^2 + \frac{v_s}{K_s} (\partial_x \varphi_s)^2 \nonumber \\*
&+ d_+ \partial_x \varphi_c \partial_x \varphi_s + d_- \partial_x \vartheta_c \partial_x \vartheta_s \bigg] ,
\end{align}
where
\begin{subequations}
\begin{align}
v_{c(s)} &\equiv [(v \pm w)^2 - (c \pm c_0)^2]^{1/2} , \\*
K_{c(s)} &\equiv \left[ \frac{v \pm w - c \mp c_0}{v \pm w + c \pm c_0} \right]^{1/2} , \\*
d_{\pm} &\equiv 2u \pm 2b .
\end{align}
\end{subequations}

\clearpage

\section{Construction of $\infty$-stable (absolutely $\infty$-stable) Luttinger liquid phases with $N \geq 23$ ($N \geq 52$), following Plamadeala \textit{et al.} (2014)}
\label{app:plamadeala_construction}

In this section we review the approach introduced in Ref.~\cite{Plamadeala2014} to construct $\infty$-stable and absolutely $\infty$-stable phases.
Note that this construction, while elegant, is not necessarily optimal, so that $\infty$-stable or absolutely $\infty$-stable phases with fewer channels than the ones constructed below may exist.

Consider the field redefinition $\phi_I = W_{IJ} \tilde{\phi}_J$, where $W \in SL(2N,\mathbb{Z})$, the group of $2N \times 2N$ matrices with integer entries and determinant 1.
This transformation permutes the integer vectors labelling vertex operators:
\beq
\mathcal{O}_\mathbf{m} = e^{i m_I \phi_I} = e^{i \tilde{m}_I \tilde{\phi}_I} ,
\eeq 
where $\tilde{\mathbf{m}} = W^T \mathbf{m} \in \mathbb{Z}^{2N}$.
Meanwhile, the fixed-point action $S$ [Eq.~(2) of the main text], written in terms of the $\tilde{\phi}$ fields, reads
\beq
\label{seq:S_LL}
S = \frac{1}{4\pi} \int dt \, dx \, \Big[ \tilde{K}_{IJ} \partial_t \tilde{\phi}_I \partial_x \tilde{\phi}_J - \tilde{V}_{IJ} \partial_x \tilde{\phi}_I \partial_x \tilde{\phi}_J \Big] ,
\eeq
where $\tilde{K} = W^T K W$ and $\tilde{V} = W^T V W$. 
The conformal spin of the operator $\mathcal{O}_\mathbf{m} = e^{i \tilde{m}_I \tilde{\phi}_I}$ is easily seen to be
\beq
K(\mathbf{m}) = \tfrac{1}{2} \tilde{\mathbf{m}}^T \tilde{K}^{-1} \tilde{\mathbf{m}} .
\eeq
Its scaling dimension is
\beq
\Delta(\mathbf{m}) = \tfrac{1}{2} \tilde{\mathbf{m}}^T \tilde{A}^T \! \tilde{A} \, \tilde{\mathbf{m}} ,
\eeq
where $\tilde{A} \in GL(2N,\mathbb{R})$ simultaneously diagonalizes $\tilde{K}$ and $\tilde{V}$; $\tilde{A} \tilde{K} \tilde{A}^T = K$, $\tilde{A} \tilde{V} \tilde{A}^T = \text{diag}(\tilde{u}_i)$.

Now assume that $\tilde{K}$ and $\tilde{V}$ are both block-diagonal:
\begin{subequations}
\begin{align}
\tilde{K} &= -\tilde{K}_R \oplus \tilde{K}_L , \\*
\tilde{V} &= \tilde{V}_R \oplus \tilde{V}_L ,
\end{align}
\end{subequations}
with $\tilde{K}_{\nu}$ and $\tilde{V}_{\nu}$ positive definite ($\nu = R/L$).
Then we can take $\tilde{A} = \tilde{Q}_R \tilde{K}_R^{-1/2} \oplus \tilde{Q}_L \tilde{K}_L^{-1/2}$, where $\tilde{Q}_{\nu} \in SO(N)$ diagonalizes $\tilde{K}_{\nu}^{-1/2} \tilde{V}_{\nu} \tilde{K}_{\nu}^{-1/2}$; it follows that
\beq
\tilde{A}^T \! \tilde{A} = \tilde{K}_R^{-1} \oplus \tilde{K}_L^{-1} .
\eeq
Thus, if $\tilde{K}$ and $\tilde{V}$ are both block-diagonal, the conformal spin and the scaling dimension are given by
\begin{subequations}
\begin{align}
K(\mathbf{m}) &= \tilde{\Delta}^L_\mathbf{m} - \tilde{\Delta}^R_\mathbf{m} , \\*
\Delta(\mathbf{m}) &= \tilde{\Delta}^L_\mathbf{m} + \tilde{\Delta}^R_\mathbf{m} ,
\end{align}
\end{subequations}
where
\begin{subequations}
\begin{align}
\tilde{\Delta}^R_\mathbf{m} &\equiv \tfrac{1}{2} \tilde{\mathbf{m}}_R^T \, \tilde{K}_R^{-1} \tilde{\mathbf{m}}_R , \\*
\tilde{\Delta}^L_\mathbf{m} &\equiv \tfrac{1}{2} \tilde{\mathbf{m}}_L^T \, \tilde{K}_L^{-1} \tilde{\mathbf{m}}_L
\end{align}
\end{subequations}
are the \emph{right} and \emph{left scaling dimensions} of the operator. 
Here, we have split $\tilde{\mathbf{m}} = (\tilde{\mathbf{m}}_R, \tilde{\mathbf{m}}_L)$, with $\tilde{\mathbf{m}}_{R/L} \in \mathbb{Z}^N$.

By construction, $\tilde{K}_{\nu}$ ($\nu = R/L$) is a positive-definite integer matrix with determinant 1, and so the same is true of its inverse.
Thus, $\tilde{K}^{-1}_{\nu}$ can be regarded as a \emph{Gram matrix} of an $N$-dimensional \emph{unimodular integral lattice} $\tilde{\Gamma}_{\nu}$ with positive definite inner product. 
Concretely, one can take the columns of $\tilde{K}^{-1/2}_{\nu}$ to form a basis for $\tilde{\Gamma}_{\nu}$, so that the lattice vectors are $\tilde{\mathbf{v}}_{\nu} = \tilde{K}_{\nu}^{-1/2} \tilde{\mathbf{m}}_{\nu}$, $\tilde{\mathbf{m}}_{\nu} \in \mathbb{Z}^N$. 
The right/left scaling dimensions are equal to half the norm-squared of these lattice vectors,
\beq
(\tilde{\Delta}^R_\mathbf{m}, \tilde{\Delta}^L_\mathbf{m}) = (\tfrac{1}{2}\abs{\tilde{\mathbf{v}}_R}^2, \tfrac{1}{2} \abs{\tilde{\mathbf{v}}_L}^2) .
\eeq
Non-chiral operators have $\tilde{\Delta}^R_\mathbf{m} = \tilde{\Delta}^L_\mathbf{m}$ and hence $\abs{\tilde{\mathbf{v}}_R} = \abs{\tilde{\mathbf{v}}_L}$. Thus, if all nonzero lattice vectors in $\tilde{\Gamma}_R$ or in $\tilde{\Gamma}_L$ have norm-squared $> 2$ (i.e.~if at least one of the two lattices is ``non-root''),  then the corresponding Luttinger liquid phase is $\infty$-stable.
There are of course \emph{chiral} operators for which only one of $\tilde{\mathbf{v}}_R$ or $\tilde{\mathbf{v}}_L$ is nonzero. Therefore, to obtain an \emph{absolutely $\infty$-stable} phase, the lattices $\tilde{\Gamma}_{R/L}$ must both have minimum norm-squared $> 4$.

Unimodular integral lattices are \emph{self-dual}, so $\tilde{K}_{\nu}$ is also a Gram matrix of $\tilde{\Gamma}_{\nu}$ (possibly with respect to a different basis).
Therefore, $\tilde{K} = -\tilde{K}_R \oplus \tilde{K}_L$ is a Gram matrix of the unimodular integral lattice $\tilde{\Gamma}_R \oplus \tilde{\Gamma}_L$ of signature $(N,N)$. Conjugating the Gram matrix $\tilde{K}$ by $W \in SL(2N,\mathbb{Z})$ corresponds merely to a basis change in this lattice. Thus, $\tilde{\Gamma}_R \oplus \tilde{\Gamma}_L \cong \mathbb{Z}^{2N}$, the signature $(N,N)$ lattice with Gram matrix $K = -\mathbb{I}_N \oplus \mathbb{I}_N$.

Let us summarize what we have accomplished so far. We have reduced the construction of $\infty$-stable (absolutely $\infty$-stable) phases of an $N$-channel Luttinger liquid to the identification of $N$-dimensional unimodular integral lattices $\tilde{\Gamma}_{R/L}$ with minimum norm-squared $> 2$ ($> 4$), subject to the constraint that $\tilde{\Gamma}_R \oplus \tilde{\Gamma}_L \cong \mathbb{Z}^{2N}$ as a lattice of signature $(N,N)$.

We now make use of two mathematical facts.
The first fact is that there is a unique signature $(N,N)$ unimodular lattice of each parity (even/odd), where an integral lattice is \emph{even} if the norm-squared of all lattice vectors is an even integer, and is \emph{odd} otherwise \scite{Conway1999}. 
The lattice $\mathbb{Z}^{2N}$ with Gram matrix  $K = -\mathbb{I}_N \oplus \mathbb{I}_N$ is clearly odd. 
Thus, $\tilde{\Gamma}_R \oplus \tilde{\Gamma}_L \cong \mathbb{Z}^{2N}$ if (and only if) at least one of $\tilde{\Gamma}_{R/L}$ is odd. 

The second fact is that, for any positive integer $\mu$, there exists an $N$-dimensional positive definite unimodular lattice whose shortest nonzero vector has $\abs{\mathbf{v}}^2 = \mu$ \scite{Milnor1973}. The required dimension $N$ increases with $\mu$; a theorem of Rains and Sloane \scite{Rains1998} states that
\beq
\label{seq:lattice_bound}
\mu \leq 2 \lfloor N/24 \rfloor + 2,
\eeq
unless $N = 23$, in which case $\mu \leq 3$.
Here $\lfloor x \rfloor$ denotes the integer part of $x$ (i.e.~$x$ rounded down).
Thus, to obtain $\mu = 3$ requires $N \geq 23$, and to obtain $\mu = 5$ requires $N \geq 48$.

In $N=23$ dimensions, the \emph{shorter Leech lattice} $\Lambda_{23}$ has minimum norm-squared $\mu = 3$.
Correspondingly, there is an $\infty$-stable $23$-channel Luttinger liquid with $\tilde{\Gamma}_R = \tilde{\Gamma}_L = \Lambda_{23}$, dubbed the ``symmetric shorter Leech liquid'' \cite{Plamadeala2014}.
The ``symmetric'' modifier distinguishes this phase from the ``\emph{asymmetric} shorter Leech liquid'' which has $\tilde{\Gamma}_R = \Lambda_{23}$ and $\tilde{\Gamma}_L = \mathbb{Z}^{23}$, and which is also $\infty$-stable.
These phases are discussed in more detail in Ref.~\cite{Plamadeala2014}, and the remarkable transport properties of the latter were analyzed in Ref.~\scite{Plamadeala2016}.

In $N=52$ dimensions, the lattice $G_{52}$ has $\mu = 5$ \scite{Gaborit2004}, and there is a corresponding \emph{absolutely $\infty$-stable} $52$-channel Luttinger liquid with $\tilde{\Gamma}_R = \tilde{\Gamma}_L = G_{52}$.

\section{Sphere packing bounds and the non-existence of absolutely $\infty$-stable phases for $N < 11$}

The \emph{sphere packing problem} \scite{Conway1999} is to find the densest possible packing of non-overlapping spheres into $\mathbb{R}^n$.
The \emph{density} of a packing is the fraction of space that is contained inside the spheres.
Given any lattice $\Gamma \subset \mathbb{R}^n$, we can obtain an associated sphere packing by placing spheres at each lattice point, with radii equal to half the length of the shortest lattice vector.
If $\Gamma$ has a unit cell of volume $\Omega$ and shortest nonzero vector of length $2r$, then the density of the associated packing, $d_{\Gamma}$, equals the volume of an $n$-ball of radius $r$, divided by $\Omega$:
\beq
d_{\Gamma} = \frac{1}{\Omega} \, \frac{\pi^{n/2} \, r^n}{\Gamma(n/2+1)} .
\eeq
Hence, upper bounds on the density of sphere packings in $\mathbb{R}^n$ yield upper bounds on the length, $2r$, of the shortest nonzero vector in $\Gamma$.

For an $N$-channel Luttinger liquid, the scaling dimensions of bosonic vertex operators are given by $\Delta(\mathbf{m}) = \frac{1}{2} \norm{A \mathbf{m}}^2$, with $A \in SO(N,N)$ and $\mathbf{m} \in D_{2N}$, the ``checkerboard lattice'' $D_{2N} \equiv \{ \mathbf{m} \in \mathbb{Z}^{2N} : \, \abs{\mathbf{m}} \in 2\mathbb{Z} \}$.
$D_{2N}$ has unit cell volume $\Omega = 2$.
Since $\det{A} = 1$, the same holds for the deformed lattice $\Gamma \equiv A D_{2N} \subset \mathbb{R}^{2N}$.
Absolute $\infty$-stability requires every nonzero vector in $\Gamma$ to have norm-squared $> 4$, which corresponds to $r > 1$.
Thus, the corresponding sphere packing would have density
\beq
d_{\Gamma} > \frac{1}{2} \, \frac{\pi^N}{\Gamma(N+1)} .
\eeq
For $N < 11$, this contradicts known upper bounds on the density of sphere packings \cite{Cohn2003}.
Hence, absolutely $\infty$-stable phases cannot exist with $N < 11$ channels.

% Bibliography for Supplement
% =======================================================================

\makeatletter
\renewcommand\@biblabel[1]{[S#1]}
\makeatother

\end{document}